\documentclass[a4paper,UKenglish,cleveref, autoref, thm-restate, numberwithinsect]{lipics-v2021}
\hideLIPIcs 
\nolinenumbers
\usepackage{macros}

\title{Tight Lower Bounds for Problems Parameterized by Rank-width}

\author{Benjamin Bergougnoux}{University of Warsaw, Poland}{benjamin.bergougnoux@mimuw.edu.pl}{0000-0002-6270-3663}{}
\author{Tuukka Korhonen}{University of Bergen, Norway}{tuukka.korhonen@uib.no}{0000-0003-0861-6515}{Supported by the Research Council of Norway via the project
BWCA (grant no. 314528).}
\author{Jesper Nederlof}{Utrecht University, The Netherlands}{j.nederlof@uu.nl}{0000-0003-1848-0076}{Supported by the project CRACKNP that has received funding from the European Research Council (ERC) under the European Union’s Horizon 2020 research and innovation programme (grant agreement No 853234)}

\authorrunning{B.\,Bergougnoux, T.\,Korhonen and J.\,Nederlof}


\acknowledgements{This work was initiated while the authors attended the “2022 Advances in Parameterized Graph Algorithms” workshop.}


\begin{document}

\maketitle

\begin{abstract}
We show that there is no $2^{o(k^2)} n^{O(1)}$ time algorithm for \textsc{Independent Set} on $n$-vertex graphs with rank-width $k$, unless the Exponential Time Hypothesis (ETH) fails.
Our lower bound matches the $2^{O(k^2)} n^{O(1)}$ time algorithm given by Bui{-}Xuan, Telle, and Vatshelle~[Discret.~Appl.~Math.,~2010] and it answers the open question of Bergougnoux and Kant{\'{e}}~[SIAM J. Discret. Math.,~2021].
We also show that the known $2^{O(k^2)} n^{O(1)}$ time algorithms for \textsc{Weighted Dominating Set}, \textsc{Maximum Induced Matching} and \textsc{Feedback Vertex Set} parameterized by rank-width $k$ are optimal assuming ETH.
Our results are the first tight ETH lower bounds parameterized by rank-width that do not follow directly from lower bounds for $n$-vertex graphs. 

\end{abstract}
 \section{Introduction}
Decompositions of graphs and their associated width parameters are a popular approach for solving NP-hard graphs problems. 
Several fundamental graph problems, like \textsc{Independent Set}, are known to be fixed-parameter tractable parameterized by width parameters like the treewidth ($\tw$) of the input graph~\cite{DBLP:journals/dam/ArnborgP89}.
While treewidth is the most prominent width parameter, it has limited applicability as it can be bounded only for sparse graphs. 
To capture the tractability of problems on structured dense graph classes, like cographs and distance-hereditary graphs, the parameter clique-width ($\cw$) was introduced by Courcelle, Engelfriet, and Rozenberg~\cite{DBLP:journals/jcss/CourcelleER93}.
Every graph with treewidth $\tw$ has clique-width $\cw \le O(2^{\tw})$, and for example complete graphs have unbounded treewidth but constant clique-width.
Many fixed-parameter tractability results parameterized by treewidth generalize to parameterization by clique-width~\cite{DBLP:journals/mst/CourcelleMR00}, and in particular \textsc{Independent Set} can be solved in time $2^{\cw} n^{O(1)}$ given a decomposition witnessing clique-width at most $\cw$.

The width parameter rank-width ($\rw$) was introduced by Oum and Seymour~\cite{DBLP:journals/jct/OumS06} in order to obtain a fixed-parameter approximation algorithm for computing the clique-width.
In particular, they showed that $\rw \le \cw \le 2^{\rw+1}-1$ and rank-width can be 3-approximated in time $8^{\rw} n^{O(1)}$, implying an exponential $2^{O(\cw)}$-approximation for clique-width within the same time complexity.
Even though multiple improvements for computing rank-width have been given~\cite{DBLP:conf/stoc/FominK22,DBLP:journals/siamcomp/HlinenyO08,DBLP:journals/siamdm/JeongKO21,DBLP:journals/talg/Oum08}, the only known way to approximate clique-width remains via computing rank-width.

In the past decade, the focus of the study of algorithms on graph decompositions has shifted from complexity classification into establishing fine-grained bounds on the time complexity as a function of the parameter~\cite{DBLP:journals/algorithmica/BergougnouxKK20,DBLP:journals/tcs/BroersmaGP13,DBLP:journals/jacm/CyganKN18,DBLP:journals/talg/CyganNPPRW22,DBLP:journals/siamcomp/FominGLS14,DBLP:journals/talg/FominGLSZ19,DBLP:conf/stacs/GroenlandMNS22,DBLP:journals/tcs/JansenN19,DBLP:journals/siamdm/Lampis20,DBLP:journals/talg/LokshtanovMS18,DBLP:journals/tcs/OumSV14}, under either the Exponential Time Hypothesis (ETH) or the Strong Exponential Time Hypothesis (SETH)~\cite{ImpagliazzoP01}.
For example, Lokshtanov, Marx, and Saurabh~\cite{DBLP:journals/talg/LokshtanovMS18} showed that assuming SETH, \textsc{Independent Set} cannot be solved in time $(2-\varepsilon)^{\pw} n^{O(1)}$ parameterized by path-width $(\pw)$ for any constant $\varepsilon>0$, which also translates into a tight lower bound of $(2-\varepsilon)^{\cw} n^{O(1)}$ parameterized by clique-width because of the relation $\cw \le \pw+2$~\cite{DBLP:journals/siamdm/FellowsRRS09}.
Lampis~\cite{DBLP:journals/siamdm/Lampis20} showed that for every constant $k \ge 3$, the optimal time complexity of $k$-coloring parameterized by clique-width is $(2^k-2)^{\cw} n^{O(1)}$ assuming SETH.

Even though fine-grained lower bounds parameterized by clique-width have been intensively studied~\cite{DBLP:journals/algorithmica/BergougnouxKK20,DBLP:journals/tcs/BroersmaGP13,DBLP:journals/siamcomp/FominGLS14,DBLP:journals/talg/FominGLSZ19,DBLP:journals/siamdm/Lampis20}, less attention has been given to fine-grained lower bounds parameterized by rank-width.
As the only known way to compute clique-width is via computing rank-width and using the constructive version of the inequalities $\rw \le \cw \le 2^{\rw+1}-1$, it could be argued that fine-grained bounds parameterized by rank-width have more significance than bounds parameterized by clique-width: In the end, the only known way to use clique-width in its full generality is to actually use rank-width.

The lack of fine-grained lower bounds for parameterizations by rank-width in the literature could be explained by the fact that the best known upper bounds appear to require more complicated arguments than for other width parameters.
In particular, while the translation to clique-width or a straightforward dynamic programming leads to a double-exponential $2^{2^{O(\rw)}} n^{O(1)}$ time algorithm for \textsc{Independent Set} parameterized by rank-width, Bui{-}Xuan, Telle, and Vatshelle showed in 2010 that surprisingly this is not optimal, giving a $2^{O(\rw^2)} n^{O(1)}$ time algorithm by exploiting the algebraic properties of rank-width~\cite{Bui-XuanTV10}.
It was asked by Bergougnoux and Kant{\'{e}}~\cite{BergougnouxK21} whether this algorithm could be shown to be optimal assuming ETH, and by Vatshelle~\cite{vatshelle:thesis} whether a $2^{O(\rw)} n^{O(1)}$ time algorithm exists.

In this paper, we show that assuming ETH, the $2^{O(\rw^2)} n^{O(1)}$ time algorithm for \textsc{Independent Set} by Bui{-}Xuan, Telle, and Vatshelle is optimal.
We show in fact a slightly more general result, using parameterization by linear rank-width ($\lrw$), which is a path-like version of rank-width whose value is at least the rank-width, i.e., $\rw \le \lrw$.

\begin{theorem}\label{thm:main}
Unless ETH fails, there is no $2^{o(\lrw^2)} n^{O(1)}$ time algorithm for \textsc{Independent Set}, where $\lrw$ is the linear rank-width of the input graphs.
\end{theorem}

Unlike for the fine-grained lower bounds parameterized by clique-width, for our result the matching upper bound holds even without the assumption that the decomposition is given because rank-width can be $3$-approximated in time $O(8^{\rw} n^4)$~\cite{DBLP:journals/talg/Oum08}.

Theorem~\ref{thm:main} is the first ETH-tight lower bound parameterized by rank-width that does not follow directly from a lower bound for $n$-vertex graphs and the relation $\rw \le n$.
Tight bounds of the latter type are known for the problems of finding a largest induced subgraph with odd vertex degrees and for partitioning a graph into a constant number of such induced subgraphs.
In particular, these problems admit $2^{O(\rw)} n^{O(1)}$ time algorithms but cannot be solved in time $2^{o(n)}$ unless ETH fails~\cite{DBLP:journals/algorithmica/BelmonteS21}.

Algorithms with time complexity $2^{O(\rw^2)} n^{O(1)}$ were given for \textsc{Dominating Set} and \textsc{Maximum Induced Matching} by Bui{-}Xuan, Telle, and Vatshelle~\cite{Bui-XuanTV10,Bui-XuanTV13} and for \textsc{Feedback Vertex Set} by Ganian and Hlinen{\'{y}}~\cite{GanianH10}.
We extend our lower bound construction to show that these algorithms for \textsc{Maximum Induced Matching} and \textsc{Feedback Vertex Set} and a weighted variant of the algorithm for \textsc{Dominating Set} are optimal assuming ETH.


\begin{theorem}\label{thm:second}
Unless ETH fails, there is no $2^{o(\lrw^2)} n^{O(1)}$ time algorithm for \textsc{Weighted Dominating Set}, \textsc{Maximum Induced Matching}, or \textsc{Feedback Vertex Set}, where $\lrw$ is the linear rank-width of the input graphs.
\end{theorem}

\subparagraph{Boolean-width.}
Boolean-width ($\bw$) is a width-parameter introduced by Bui{-}Xuan, Telle, and Vatshelle~\cite{DBLP:journals/tcs/Bui-XuanTV11}.
It is defined on branch decompositions over the vertex set $V(G)$ with the cut function $\bw(A,B) = \log_2\abs{\{ N(X)\cap B\mid X\subseteq A\}}$, which naturally leads to algorithms with time complexity $2^{O(\bw)} n^{O(1)}$ for local problems when a branch decomposition with Boolean-width $\bw$ is given.
Bui{-}Xuan, Telle, and Vatshelle proved that $\log_2 \rw \leq \bw \leq O(\rw^2)$ and asked as an open question whether the upper bound $\bw \leq O(\rw^2)$ is tight~\cite{DBLP:journals/tcs/Bui-XuanTV11}.
Since \textsc{Independent Set} can be solved in time $2^{O(\bw)} n^{O(1)}$ given a branch decomposition with Boolean-width $\bw$, \cref{thm:main} gives evidence that there are graphs whose Boolean-width is quadratic in rank-width.
We show that indeed a small variant of a construction used for \cref{thm:main} can be used to give a quadratic separation between Boolean-width and rank-width, answering the question of Bui{-}Xuan, Telle, and Vatshelle.

\begin{theorem}\label{thm:sep}
There are graphs with rank-width $k$ and Boolean-width $\Omega(k^2)$ for arbitrarily large $k$.
\end{theorem}

\subparagraph{Our Method.} We briefly describe our main technical ideas for Theorem~\ref{thm:main}, as the other ideas are similar. 

Our starting point is the $2^{O(\rw^2)}n^{O(1)}$ time algorithm for \textsc{Independent Set} from~\cite{Bui-XuanTV10}. Intuitively, the algorithm can be thought of as normal dynamic programming over tree decompositions where we have table entries for each subset of a bag of the tree-decomposition (which is a separator in the graph). The twist however is that the size of the separator may be large, but instead we are only guaranteed that the rank (over $\mathbb{F}_2$) of the incidence matrix of the cut between the separator and the remainder of the graph is small. The crucial observation from~\cite{Bui-XuanTV10} is that we do not need to know the exact subset of the separator of vertices selected in the independent set, but only its set of neighbors across the cut, and that such a neighborhood can be described by a (row- or column-) \emph{subspace} of the mentioned incidence matrix. Since there are only $2^{O(\rw^2)}$ such subspaces, the runtime follows.

To turn this encoding idea into a reduction from \textsc{$3$-CNF-SAT} to \textsc{Independent Set} on graphs of rank-width $\rw$ (and get an ETH-tight lower bound), we first show this description cannot be further shortened: Two vertex sets of the separator that describe different subspaces in fact will have different sets of neighbors accross the cut. This allows us to design a ``copy gadget'': Once locally a certain vertex subset is chosen to be in the independent set, this is has to be copied in various places throughout the graph (such that we can check a clause of the \textsc{$3$-CNF} per one such location).
Furthermore, there is a simple inductive construction of subspaces that enables us to directly encode assignments of $\Omega(\rw^2)$ variables of an \textsc{$3$-CNF} formula into a subspace of $\mathbb{F}_2^{\rw}$. This allows us to get a tight bound.

\subparagraph{Organization.}
The rest of this paper is organized as follows.
In Section~\ref{sec:prel} we set up notation.
In Section~\ref{sec:fullrank} we present structural results on the maximal unique cut with rank-width $k$ (that we call ``the universal $k$-rank cut''). 
Subsequently, Section~\ref{sec:IS} builds upon these results to prove~\cref{thm:main}.
In \cref{sec:MIMFVS} we prove the lower bounds for \textsc{Maximum Induced Matching} and \textsc{Feedback Vertex Set}.
We prove the lower bound for \textsc{Weighted Dominating Set} in \cref{sec:wds} and we prove \cref{thm:sep} in \cref{sec:rel}.
We provide a brief conclusion in Section~\ref{sec:conc}.

\section{Notation}
\label{sec:prel}
Given two integers $i,j$ such that $1\leq i\leq j$, we denote by $[i,j]$ the set of integer $\{i,i+1,\dots,j\}$ and by $[i]$ the set $\{1,2,\dots,i\}$. 
The size of a set $V$ is denoted by $|V|$ and its power set is denoted by $2^V$. 

\subparagraph{Graphs.}
Our graph terminology is standard and we refer to \cite{Diestel12}.  The vertex set of a graph $G$ is denoted by $V(G)$ and its edge set by $E(G)$.
For every vertex set $X\subseteq V(G)$, when the underlying graph is clear from context, we denote by $\comp{X}$ the set $V(G)\setminus X$.
An edge between two vertices $x$ and $y$ is denoted by $xy$ or $yx$. 
The set of vertices that are adjacent to $x$ is denoted by $N_G(x)$. For a set $U\subseteq V(G)$, we define  $N_G(U):=\bigcup_{x\in U}N_G(x) \setminus U$. If the underlying graph is clear, then we may remove $G$ from the subscript.
Two distinct vertices $u,v \in V(G)$ are twins if $N(v) \setminus \{u\} = N(u) \setminus \{v\}$.
They are true twins if $uv \in E(G)$ and false twins if $uv \notin E(G)$.

The subgraph of $G$ induced by a subset $X$ of its vertex set is denoted by $G[X]$.
For two disjoint subsets $X$ and $Y$ of $V(G)$, we denote by $G[X,Y]$ the bipartite graph with vertex set $X\cup Y$ and edge set $\{xy \in E(G)\mid x\in X \text{ and } \ y\in Y \}$.

\subparagraph{Problem Statements.}
An \emph{independent set} is a set of vertices that induces an edgeless graph.
A \emph{matching} is a set of edges having no common endpoint and an \emph{induced matching} is a matching $M$ where every pair of edges of $M$ do not have a common adjacent edge in $G$.
Given a graph, the problems \textsc{Independent Set} and \textsc{Maximum Induced Matching} ask for respectively an independent set and an induced matching of maximum size.

A \emph{feedback vertex set} is the complement of a set of vertices inducing a forest (i.e. acyclic graph).
A set $D\subseteq V(G)$ \emph{dominates} a set $U\subseteq V(G)$ if every vertex in $U$ is either in $D$ or is adjacent to a vertex in $D$. A \emph{dominating set} of $G$ is a set that dominates $V(G)$.
Given a graph, the problems \textsc{Dominating Set} and \textsc{Feedback Vertex Set}  ask respectively for a dominating set and a feedback vertex set of minimum size.

Given a graph $G$ with a weight function $w : V(G)\to \bN$, 
the problem \textsc{Weighted Independent Set} (resp. \textsc{Weighted Dominating Set}) asks for an independent set of maximum weight (resp. dominating set of minimum weight), where the weight of a set $X\subseteq V(G)$ is $\sum_{x\in X} w(x)$.

\subparagraph{Width parameters.}
Let $V$ be a finite set with $|V| \ge 3$, and $f$ a function $f : 2^V \rightarrow \mathbb{Z}_{\ge 0}$ so that $f(\emptyset) = 0$ and $f(A) = f(\comp{A})$ for all $A \subseteq V$.
A branch decomposition of $f$ is a tree whose all internal nodes have degree 3 and whose leaves are bijectively mapped to $V$.
Observe that every edge of a branch decomposition of $f$ corresponds to a bipartition $(A,\comp{A})$ of $V$ given by the leaves on different sides of the edge.
The width of the decomposition is the maximum value of $f(A)$ over the bipartitions $(A,\comp{A})$ corresponding to the edges, and the branch-width of $f$ is the minimum width of a branch decomposition of $f$.
The width of a permutation $\sigma$ of $V$ is the maximum value of $f(A)$ over prefixes $A$ of the permutation.
The linear branch-width of $f$ is the minimum width of a permutation of $V$.
Notice that the branch-width of $f$ is always at most the linear branch-width of $f$, in particular linear branch-width corresponds to a restriction of branch-width where we demand the tree to be a caterpillar.

Let $G$ be a graph and $A,B \subseteq V(G)$ two disjoint sets of vertices.
We define $M_G(A,B)$ to be the $|A| \times |B|$ 0-1 matrix representing the bipartite graph $G[A,B]$.
The cut-rank $\rw(A,B)$ between $A$ and $B$ is defined as the $\mathbb{F}_2$-rank of $M_G(A,B)$.
For a set of vertices $A \subseteq V(G)$ we define $\rw(A) = \rw(A, \comp{A})$.
The rank-width of a graph $G$ is the branch-width of the cut-rank function $\rw$ and the linear rank-width of $G$ is the linear branch-width of $\rw$.
The Boolean-rank between $A$ and $B$ is defined as $\bw(A,B) = \log_2|\{N(X)\cap B \mid X \subseteq A\}|$, and we define $\bw(A) = \bw(A, \comp{A})$.
The Boolean-width of a graph $G$ is the branch-width of $\bw$.

We will use the following lemma about the cut-rank.

\begin{lemma}
\label{lem:cutcomb}
Let $A,B$ be disjoint subsets of $V(G)$ and $S \subseteq V(G)$.
It holds that $\rw(A,B) \le \rw(A \cap S, B)+\rw(A \setminus S, B)$.
\end{lemma}
\begin{proof}
The edges of $G[A \cap S, B]$ and $G[A \setminus S, B]$ are disjoint, and therefore $M_G(A,B)$ can be written as the sum of (appropriately permuted) $M_G(A \cap S, B)$ and $M_G(A \setminus S, B)$.
The lemma follows from the fact that the rank of a sum of matrices is at most the sum of the ranks of the matrices.
\end{proof}

\section{Structural Results on the Universal \texorpdfstring{$k$}{k}-Rank Cut}
\label{sec:fullrank}
In this section, we study the \emph{universal $k$-rank cut} which is the unique inclusion-wise maximal cut of a given rank $k$ with no twin vertices.
This cut was used in \cite{Bui-XuanTV10} to give a non-algebraic definition of rank-width and in \cite{DBLP:journals/tcs/Bui-XuanTV11} to prove that some cuts can have rank-width $k$ and Boolean-width $\Omega(k^2)$ for arbitrary large $k$.

\begin{definition}
For any integer $k\geq 1$, the universal $k$-rank cut $R_k$ is defined as the bipartite graph with $2 \cdot 2^k$ vertices with color classes $A^k\defeq\{a_s \mid s\subseteq [k]\}$ and $B^k\defeq\{b_t \mid t \subseteq [k]\}$ such that $a_s$ and $b_t$ are adjacent if and only if $\abs{s\cap t}$ is odd.
Given $\cS=\{s_1,\dots,s_\ell\}\subseteq 2^{[k]}$, we define  $A^k[\cS]\defeq\{a_{s_1},\dots,a_{s_\ell}\}$ and $B^k[\cS]\defeq\{b_{s_1},\dots,b_{s_\ell}\}$.
    We may omit $k$ from the superscript when it is clear from the context.
\end{definition}

The following lemma characterizes the cut-rank in terms of $R_k$.

\begin{lemma}[\cite{Bui-XuanTV10}]
Let $G$ be a graph and $A,B$ disjoint subsets of $V(G)$.
It holds that $\rw(A,B) \le k$ if and only if $G[A,B]$, after removing twins, is an induced subgraph of $R_k$.
\end{lemma}

In particular, it follows that in the graph $R_k$ we have $\rw(A^k, B^k) = k$.

In the following, we define a family $\sF_k$ of subsets of $2^{[k]}$ that will play a crucial role in our reductions.
One major property of this family will be that for every distinct $\cS_1,\cS_2$ in $\sF_k$, the sets of vertices $A^k[\cS_1]$ and $A^k[\cS_2]$ have different neighborhoods in $R_k$.

\begin{definition}\label{def:sF_k}
	Given any integer $k\geq 1$ and $\cS=\{s_1,\dots,s_\ell\}\subseteq 2^{[k]}$, we define $\cS\otimes k+1$ as the collection consisting of $\{s_1,\dots,s_\ell,\{k+1\}\}$ and all the sets $\{s_1',\dots,s_\ell'\}$ such that for every $i\in[\ell]$, we have $s_i'\in \{s_i, s_i\cup \{k+1\}\}$.
	We define $\sF_1$ as the family containing $\emptyset$ and $\{\{1\}\}$.
	For every $k\geq 1$, we define $\sF_{k+1}= \bigcup_{\cS\in \sF_{k}} \cS\otimes k+1$.
\end{definition}

For example, the collection $\sF_2$ is the union of $\emptyset\otimes 2$ and $\{\{1\}\}\otimes 2$ where $\emptyset\otimes 2= \{\emptyset , \{\{2\}\}\}$ and $\{\{1\}\}\otimes 2$ is the collection containing $\{\{1\}\}$, $\{\{1,2\}\}$ and $\{\{1\},\{2\}\}$.

It is natural to view $\mathcal{S} = \{s_1,\ldots,s_\ell\} \subseteq 2^{[k]}$  as a binary $\ell \times k$  matrix (so the $s_i$'s are interpreted as $k$-dimensional binary vectors). Then the $\otimes$ operation can be thought of as transforming a single matrix into a set of matrices as displayed in \cref{fig:matrix:interpretation}, and the family $\mathcal{F}_i$ can be thought of as the family of all different binary matrices with $k$ columns that are in row-reduced echelon form: The steps in which  we add a row are exactly the pivotal columns.

\begin{figure}[ht]
	\centering
	\includegraphics[width=\linewidth]{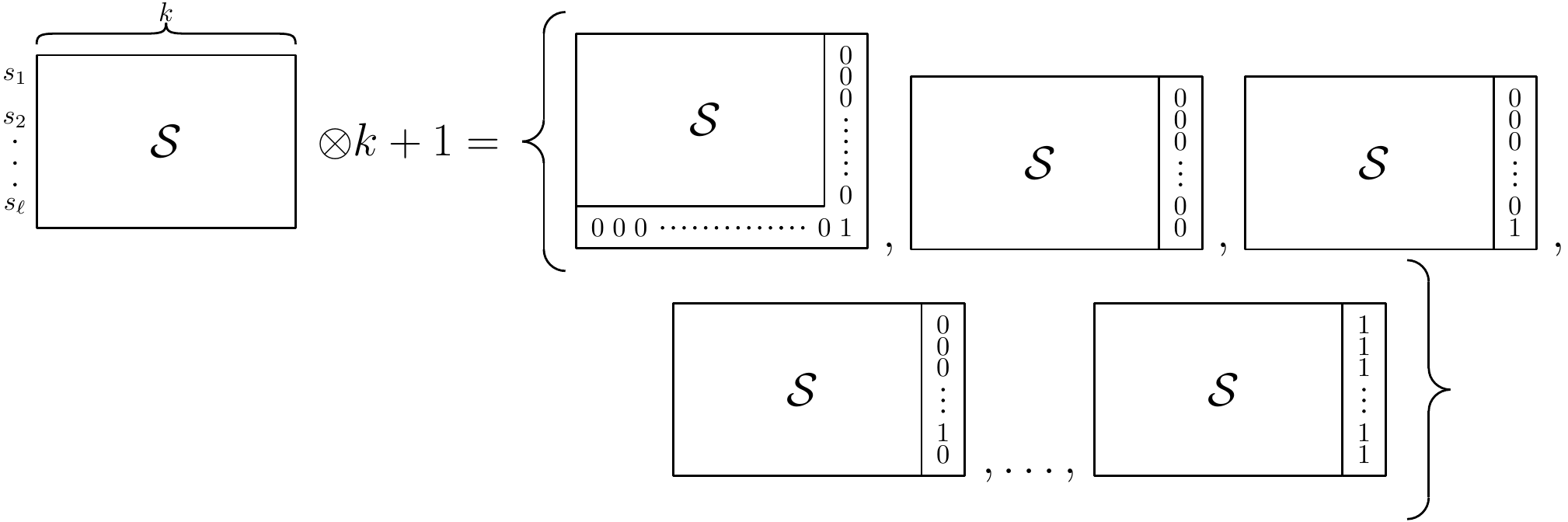}
	\caption{Illustration of the interpretation of the $\otimes$ operation as an operation on binary matrices.}
	\label{fig:matrix:interpretation}
\end{figure}

The following observation gives an alternative definition of $\sF_k$.

\begin{observation}\label{obs:sF_k}
	For every $k\geq 1$ and $\cS\subseteq 2^{[k]}$, we have $\cS=\{s_1,\dots,s_{\abs{\cS}}\}\in\sF_{k}$ if and only if there exist pairwise distinct integers $\alpha_1,\dots,\alpha_{\abs{\cS}}\in [k]$ such that for every $i\in [\abs{\cS}]$, we have $s_{i}\cap \{\alpha_1,\dots,\alpha_{|\cS|}\}=\{\alpha_i\}$ and $s_i\subseteq \{ \beta\in [k] \mid \alpha_i \leq \beta \}$.
\end{observation}
Note that the distinct integers $\alpha_1,\ldots , \alpha_{|\cal S|}$ in \Cref{obs:sF_k} are the column indices of the leading coefficients in the matrix formulation in \cref{fig:matrix:interpretation}.

We are ready to prove the two properties on the collections in $\sF_k$ that we will use in our reductions.
We first prove the distinct neighborhoods property.

\begin{lemma}\label{lem:different:neighborhood}
For every $k \ge 1$ and every pair $\cS_1, \cS_2\in \sF_k$ with $\cS_1 \neq \cS_2$ it holds that $N_{R_k}(A^k[\cS_1])\neq N_{R_k}(A^k[\cS_2])$.
\end{lemma}
\begin{proof}
	We start with the following claim.
	\begin{claim}\label{claim:private:neighbor}
		For every $k\geq 1$ and $\cS,\cX\subseteq 2^{[k]}$ with $\cX\subseteq \cS\in \sF_k$, there exists $t\subseteq [k]$ such that for every $s\in \cS$, $\abs{s \cap t}$ is odd if and only if $s\in \cX$ (i.e. $N(b_t)\cap A^k[\cS]=A^k[\cX]$).
	\end{claim}
	\begin{proof*}
        Let $\cS=\{s_1,\dots,s_\ell\}\in \sF_k$  and $\cX\subseteq \cS$.
        By \Cref{obs:sF_k}, there exists $\alpha_1,\dots,\alpha_\ell\in [k]$ such that for every $i\in [\ell]$, we have $s_{i}\cap \{\alpha_1,\dots,\alpha_{\ell}\}=\{\alpha_i\}$.
        Let $t=\{\alpha_i \mid i\in [\ell] \land s_i \in \cX\}$.
        Observe that, for every $i\in[\ell]$, we have $s_i\cap t=\{\alpha_i\}$ iff $s_i \in \cX$.
        Hence, $\abs{s_i\cap t}$ is odd iff $s_i\in \cX$ for every $i\in[\ell]$.
	\end{proof*}
	
	We are now ready to prove the lemma by induction on $k$.
	It is obviously true for $\sF_1$.
	Let $k\geq 1$ and suppose that the lemma holds for $\sF_k$.
	Let $\cS_1',\cS_2'\in \sF_{k+1}$ such that $\cS_1'\neq \cS_2'$.
	By construction of $\sF_{k+1}$, for every $i\in\{1,2\}$, there exists a unique $\cS_i\in\sF_k$ such that $\cS_i'\in \cS_i\otimes k+1$.
	
	If $\cS_1\neq \cS_2$, then by induction hypothesis, $N_{R_k}(A^k[\cS_1])\neq N_{R_k}(A^k[\cS_2])$.
	By definition of $R_{k+1}$, for each $i\in \{1,2\}$, the sets of vertices $A^{k+1}[\cS_i]$ and $A^{k+1}[\cS_i']$ have the same neighborshoods in $R_{k+1}$ when restricted to the subset $B^k \subseteq B^{k+1}$.
	We conclude that if $\cS_1\neq \cS_2$, then $N_{R_{k+1}}(A^{k+1}[\cS_1'])\neq N_{R_{k+1}}(A^{k+1}[\cS_2'])$.
	
	It remains to consider the case when $\cS_1=\cS_2=\{s_1,\dots,s_\ell\}$.
	For both $i\in \{1,2\}$, we define a set $\cX_i \defeq \{s_j \mid j\in [\ell] \land s_j \cup \{k+1\}\in \cS'_i\}$.
	Since $\cS_1'\neq \cS_2'$, we have $\cX_1\neq \cX_2$.
	By \Cref{claim:private:neighbor}, there exists $t_1\subseteq [k]$ such that $\abs{s_j\cap t_1}$ is odd if and only if $s_j\in\cX_1$ for every $j\in [\ell]$.
	We claim that $b_{t_1\cup \{k+1\}}$ is adjacent to  $A^{k+1}[\cS_2']$ but not to $A^{k+1}[\cS_1']$ in $R_{k+1}$.
	For every $j\in [\ell]$, we have $s_j\in \cX_1$ if and only if $s_j\cup\{k+1\}\in \cS_1'$, we deduce that $\abs{s\cap (t_1\cup \{k+1\})}$ is even for every $s\in \cS_1'$.
	Consequently, $b_{t_1\cup \{k+1\}}$ is not adjacent to $A^{k+1}[\cS_1']$ in $R_{k+1}$.
	As $\cX_1\neq \cX_2$, there exists $j\in [\ell]$ such that $s_j\in \cX_1 \triangle \cX_2$.
	\begin{itemize}
		\item If $s_j\in \cX_1\setminus \cX_2$, then $s_j\in \cS'_2$ and $\abs{s_j\cap (t_1\cup \{k+1\})}=\abs{s_j\cap t_1}$ is odd.
		\item If $s_j\in \cX_2\setminus \cX_1$, then $s_j\cup \{k+1\}\in \cS'_2$ and $\abs{s_j\cup \{k+1\}\cap (t_1\cup \{k+1\})}=\abs{s_j\cap t_1}+1$ is odd since $\abs{s_j\cap t_1}$ is even.
	\end{itemize}
	We deduce that $b_{t_1\cup \{k+1\}}$ is adjacent to $A^{k+1}[\cS_2']$ and conclude that $N_{R_{k+1}}(A^{k+1}[\cS_1'])\neq N_{R_{k+1}}(A^{k+1}[\cS_2'])$.
\end{proof}

We then show that the size of the neighborhood of $A^k[\cS]$ depends only on $|\cS|$.

\begin{lemma}\label{lem:neighborhood:size}
	For every $k\geq 1$ and $\cS\in \sF_k$, we have $\abs{N_{R_k}(A^k[\cS])}=2^{k}-2^{k-\abs{\cS}}$.
\end{lemma}
\begin{proof}
	Let $k\geq 1$ and $\cS\in \sF_k$.
	In this proof, we view each subset $s\subseteq [k]$ as a vector $(s_1,\dots,s_k)\in \bZ_2^k$ with $s_i =1$ if and only if $i\in s$ for every $i\in [k]$.
	Observe that, given $s,t\subseteq [k]$, the inner product $\inprod{s,t}$ of $s$ and $t$ is the sum $s_1t_1 + s_2t_2 +\dots +s_kt_k$ over $\bZ_2$ and it equals 1 if and only if $\abs{s\cap t}$ is odd. 
	
	We denote by $\inprod{\cS}$ the subspace of $\bZ_2^k$ generated by the vectors of $\cS$ and by $\inprod{\cS}^\perp \defeq \{t\subseteq [k] \mid \forall s\in \inprod{\cS}, \inprod{s,t}=0\}$ its orthogonal subspace.
	
	Observe that, for every $s_1,s_2,t\subseteq [k]$, if $\inprod{s_1,t}=\inprod{s_2,t}=0$, then $\inprod{s_1+s_2, t}=0$.
	Consequently, we have $t\in \inprod{\cS}^\perp$ if and only if $\inprod{s,t}=0$ for every $s\in \cS$.
	We deduce that \[ N_{R_k}(A^k[\cS])= B^k\setminus B^k[\inprod{\cS}^{\perp}]. \]
	By the rank-nullity theorem, we have $\dim \inprod{\cS} + \dim \inprod{\cS}^{\perp}=k$.
	By \Cref{claim:private:neighbor}, for every $s\in \cS$, there exists $t\subseteq [k]$ such that $\inprod{s,t}=1$ and $\inprod{s',t}=0$ for every $s'\in \cS \setminus \{s\}$.
	We deduce that $\dim \inprod{\cS}=\abs{\cS}$ and thus $\dim \inprod{\cS}^{\perp}=k-\abs{\cS}$.
	As $N_{R_k}(A^k[\cS])= B^k\setminus B^k[\inprod{\cS}^{\perp}]$, we conclude that $\abs{N_{R_k}(A^k[\cS_1])} = \abs{\bZ^k_2} - \abs{\inprod{\cS}^{\perp}}= 2^{k} - 2^{k-\abs{\cS}}$.
\end{proof}

\section{Reduction for Independent Set}
\label{sec:IS}

Our reduction is from the variant of \textsc{$3$-CNF-SAT} with a square number of variables.
Since, for every $n\in \bN$, there is always a square number between $n$ and $2n+1$, we can always add $O(n)$ dummy variables to an instance of \textsc{$3$-CNF-SAT} to ensure a square number of variables. Thus, we have the following easy consequence of ETH~\cite{ImpagliazzoP01}.

\begin{lemma}\label{lem:3SAT:square:ETH}
Unless ETH fails, there is no $2^{o(k^2)} (k+m)^{O(1)}$ time algorithm for \textsc{$3$-CNF-SAT} with $m$ clauses and $k^2$ variables, where $k$ is an integer.
\end{lemma}

We start from an instance $\phi$ of \textsc{$3$-CNF-SAT} with $m$ clauses $C_1,\dots,C_m$ and a set of $k^2$ variables $\var(\phi)\defeq\{v_{i,j} \mid i\in  [k], j\in [k+1,2k]\}$.
The main idea of our construction is to use the following bijection between the assignments of $\var(\phi)$ and some collections in $\sF_{2k}$.

\begin{definition}\label{def:interpretation:sF2k}
	Let $\cS\defeq\{s\subseteq [2k] \mid \abs{s\cap [k]} = 1\}$.
	For each $i\in [k]$, we denote by $\cS_i$ the set $\{s\in \cS \mid s\cap [k] = \{i\}\}$.
	For every assignment $f : \var(\phi) \to \{0,1\}$, we denote by $\cS_f$ the collection containing the sets $s_1\in \cS_1,s_2\in \cS_2,\dots, s_k\in \cS_k$, where for every $i\in [k]$
	\[
	s_i=\{i\} \cup \{ j\in [k+1,2k] \mid f(v_{i,j})=1 \}.
	\]
	%
	Given a literal $\ell\in \{v_{i,j}, \neg v_{i,j}\}$, we define $\cS_{\ell}$ as the set $\{ s\in \cS_i \mid j\in s \}$ if $\ell=v_{i,j}$ and $\{s\in \cS_i \mid j\notin s\}$ if $\ell=\neg v_{i,j}$.
\end{definition}

For example with $k=2$, the interpretation $f$ which sets the variables $v_{1,3},v_{1,4}$ and $v_{2,4}$ to true is associated with the collection $\cS_f$ containing $\{1,3,4\}$ and $\{2,4\}$.

\begin{observation}\label{obs:interpretation:literal}
	For every interpretation $f : \var(\phi) \to \{0,1\}$ and literal $\ell \in \{v_{i,j},\neg v_{i,j} \mid v_{i,j} \in \var(\phi)\}$, we have $f(\ell)=1$ if and only if $\cS_f\cap \cS_{\ell} \neq \emptyset$ if and only if $\cS_f\cap \cS_{\neg \ell}=\emptyset$.
\end{observation}

\subparagraph{High level description of the reduction.} We consider a modified version of the universal $2k$-rank cut $R_{2k}$ which is obtained by: (1)~removing the vertices that are in $A^{2k}$ but not in $A^{2k}[\cS]$ and (2)~making  $A^{2k}[\cS_i]$ a clique for every $i\in [k]$.
This way in the graph induced by $A^{2k}[\cS]$, every maximal independent set is of the form $A^{2k}[\cS_f]$ with $f$ an assignment of $\var(\phi)$.

We make $m$ copies $A_1, \ldots, A_m$ of such $A^{2k}[\cS]$ and for each $i\in [m]$, we create a simple clause gadget that is adjacent to some vertices of $A_i$.
For each $i\in [m]$, the clause gadget for $C_i$ is simply a triangle whose vertices are associated with the literals of $C_i$. For each literal $\ell$ of $C_i$, the vertex associated with $\ell$ is adjacent to the vertices in the $A_i$ associated with the sets in $\cS_{\neg\ell}$.
See \Cref{fig:isoverview} for an overview of this construction and \Cref{fig:is:clause:gadget} for an example of clause gadget.

Finally, we define a vertex-weight function (that can be emulated in the unweighted setting by adding false twins) in such a way that for any independent set $I$ of maximum weight, there exists an interpretation $f$ of $\var(\phi)$ such that $I$ contains the copies of $A^{2k}[\cS_f]$.
To this end, we actively use the fact that for all interpretations $f,g$ of $\var(\phi)$, we have $\cS_f,\cS_g\in \sF_{2k}$ and thus $A^{2k}[\cS_f]$ and $A^{2k}[\cS_g]$ have the same neighborhood in $R_{2k}$ if and only if $f=g$.

\begin{figure}[hbt]
	\centering
	\includegraphics[width=0.7\linewidth]{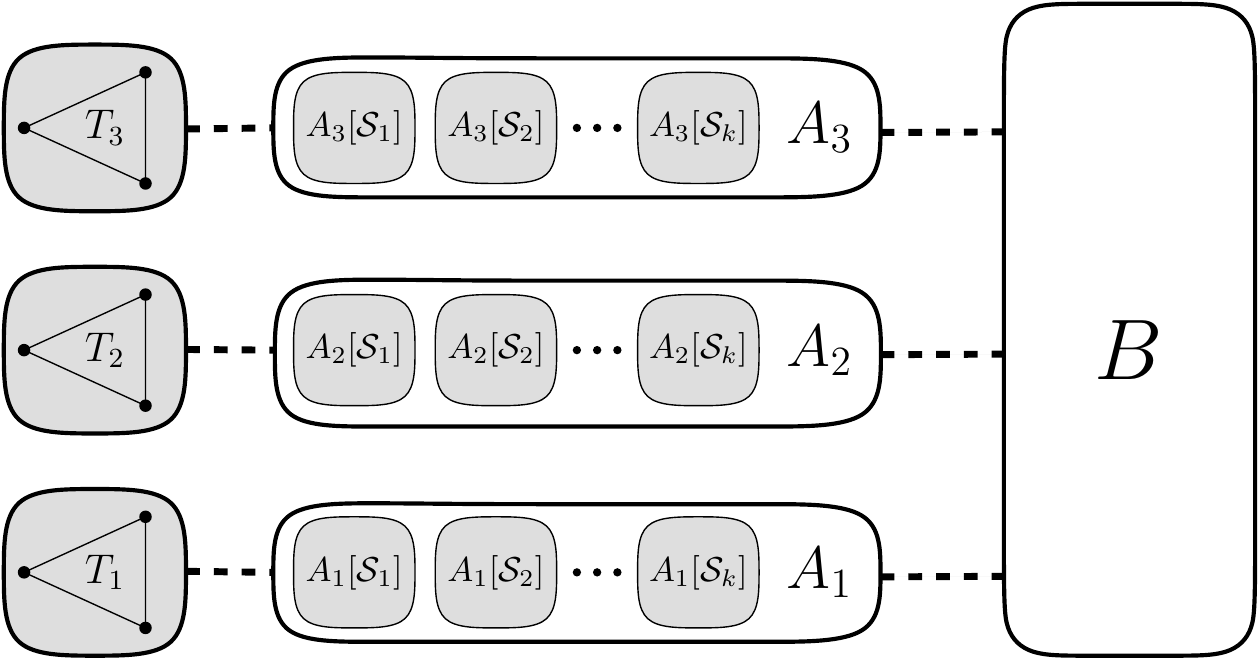}
	\caption{Overview of the reduction for \textsc{Independent Set} with $m=3$. The gray areas represent cliques and dotted lines indicates the existence of edges between two sets of vertices.}
	\label{fig:isoverview}
\end{figure}

\begin{figure}[hbt]
	\centering
	\includegraphics[width=0.8\linewidth]{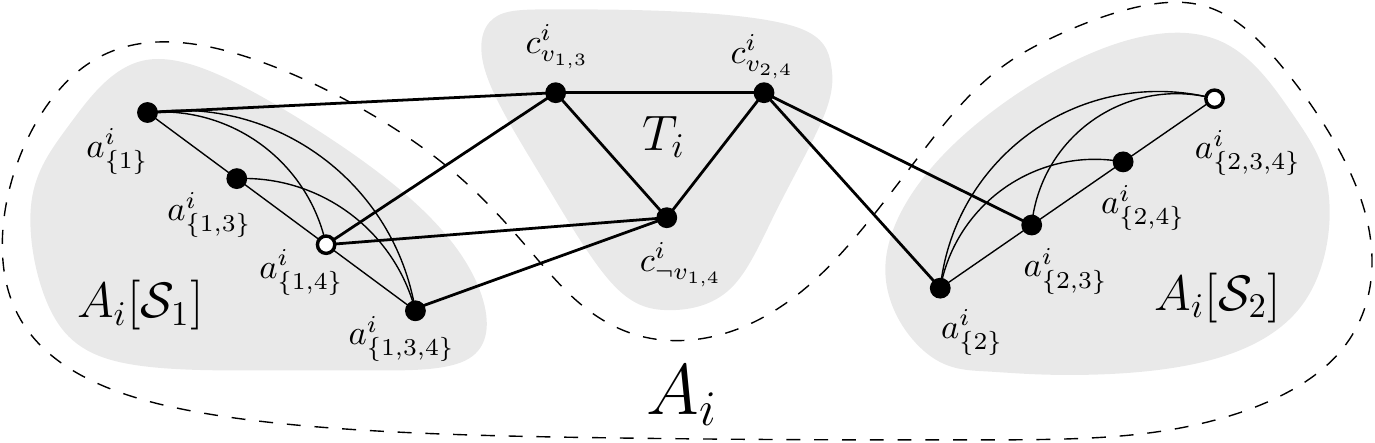}
	\caption{Example of clause gadget with $k=2$ and a clause $C_i=\{v_{1,3},\neg v_{1,4}, v_{2,4}\}$.
	The independent set containing the white-filled vertices is $A_i[\cS_f]$ for the interpretation $f$ with $f^{-1}(1)=\{v_{1,4},v_{2,3},v_{2,4}\}$ and $f^{-1}(0)=\{v_{1,3}\}$. This independent set can be extended with the vertex $c^i_{v_{2,4}}$.}
	\label{fig:is:clause:gadget}
\end{figure}

\subparagraph{The construction.}  
We construct a graph $G$ as follows.
We create $m$ copies $A_1,\dots,A_m$ of $A^{2k}[\cS]$, for each $i\in [m]$, we have $A_i\defeq \{a^i_s \mid s \in \cS\}$.
Given $\cS'\subseteq \cS$ and $i\in [m]$, we denote by $A_i[\cS']$ the set $\{a^i_s \in A_i \mid s\in \cS'\}$.
For each $i\in[m]$ and $j\in [k]$, we add edges so that $A_i[\cS_j]$ induces a clique.
Thanks to these $k$ cliques, we have the following relation between the independent sets of $G[A_i]$, the subsets of $\cS$ and the interpretations of $\var(\phi)$.

\begin{lemma}\label{lemma:IS:interpretation}
	Let $i\in [m]$ and $\cS'\subseteq \cS$. If $A_i[\cS']$ is an independent set of $G[A_i]$, then $\cS'\in\sF_{2k}$.
	Moreover, $A_i[\cS']$ is an independent set of $G[A_i]$ of size $k$ if and only if there exists an interpretation $f$ of $\var(\phi)$ such that $\cS'=\cS_f$.
\end{lemma}
\begin{proof}
	Let $i\in [m]$ and $\cS'=\{s_1,\dots,s_{r}\}\subseteq \cS$ such that $A_i[\cS']$ is an independent set of $G$.
	By construction of $G[A_i]$, we know that $A_i[\cS']$ contains at most one vertex from $A_i[\cS_j]$ for each $j\in [k]$.
	Thus, there exist pairwise distinct integers $\alpha_1,\dots,\alpha_{r}\in [k]$ such that for every $j\in [r]$, we have $\{\alpha_j\} \subseteq s_j\subseteq \{\alpha_j\}\cup [k+1,2k]$.
	By \Cref{obs:sF_k}, we deduce that $\cS'$ belongs to the collection $\sF_{2k}$.
	
	It is easy to see that $\abs{\cS'}=k$ if and only if $\cS'=\{s_1,\dots,s_k\}$ with $\{j\} \subseteq s_j \subseteq \{j\}\cup [k+1,2k]$ for every $j\in [k]$ and this is equivalent to $\cS'=\cS_f$ with $f^{-1}(1)=\{v_{i,j}\mid i\in [k] \land j\in s_i  \cap [k+1,2k] \}$.	
\end{proof}

We create $B =B^{2k}\defeq \{b_s \mid s\subseteq [2k]\}$.
For every $i\in [m]$, $s\in \cS$ and $t\subseteq [2k]$, if $\abs{s\cap t}$ is odd, we make $b_t$ adjacent to $a_s^i$. Consequently, $G[A_i, B_i]$ is isomorphic to $R_{2k}[A^{2k}[\cS], B^{2k}]$. 
We deduce the following observations from \Cref{lemma:IS:interpretation} and our results on $R_{2k}$.

\begin{observation}\label{claim:IS:neighborhood:different}
	Let $i,j\in [m]$ and $f,g$ be two assignments of $\var(\phi)$. We have $N(A_{i}[\cS_f])\cap B = N(A_{j}[\cS_g]) \cap B$ if and only if $f=g$.
\end{observation}

\begin{observation}\label{claim:IS:size}
	For every $i\in [m]$ and $\cS'\subseteq \cS$ such that $A_i[\cS']$ is an independent set of $G$, we have 
	$ \abs{  N(A_{i}[\cS']) \cap B} = 2^{2k} - 2^{2k-\abs{\cS'}}. $
\end{observation}

For every $i\in [m]$ with $C_i=\{\ell_1,\ell_2,\ell_3\}$, we create a triangle induced by a set $T_i$ of three new vertices $c^i_{\ell_1},c^i_{\ell_2}$ and $c^i_{\ell_3}$.
For each $j\in [3]$, we make $c^i_{\ell_j}$ adjacent to all the vertices in $A_i[\cS_{\neg \ell_j}]$.
Thanks to \Cref{obs:interpretation:literal} and these edges, each independent sets associated with an interpretation that satisfies $C_i$ can be extended with one of the vertices in $T_i$.

\begin{observation}\label{obs:IS:clause:gadget}
	For every assignment $f$ of $\var(\phi)$ and $i\in [m]$ with $C_i=\{\ell_1,\ell_2,\ell_3\}$, we have $T_i \setminus N(A_{i}[\cS_f])\neq \emptyset$ if and only if $f$ satisfies $C_i$.
\end{observation}

Finally, we define the weight function $w : V(G) \to \bN$ such each vertex $v$ in $A_1\cup \dots \cup A_m$ have weight $w(v)\defeq 2^{2k} = \abs{B}$ and all the other vertices have weight $1$.
The purpose of $w$ is to guarantee that maximum independent set of $G$ contains $k$ vertices in each $A_i$.

We are ready to prove the correctness of our reduction.

\begin{lemma}\label{lem:SAT:IS}
	If $\phi$ is a satisfiable 3-CNF-SAT formula, then $G$ admits an independent set of weight $2^{2k}km + 2^{k}+m$.
\end{lemma}
\begin{proof}
	Suppose $\phi$ admits a satisfying assignment $f$.
	For each $i\in [m]$, let $\ell^i$ be a literal of $C_i$ such that $f(\ell^i)=1$.
	Let $I$ be the set of vertices that contains (1)~the vertices in $B\setminus N(A_1[\cS_f])=\dots=B\setminus N(A_m[\cS_f])$ and (2)~for every $i\in [m]$ the vertices in $A_i[\cS_f] \cup \{ c^i_{\ell^i}\}$.
		
	Lemmas~\ref{lemma:IS:interpretation} and \ref{claim:IS:size} and Observations~\ref{obs:IS:clause:gadget} and \ref{claim:IS:neighborhood:different} imply that $I$ is an independent set of $G$.
	For every $i\in [m]$, we have $\abs{A_i\cap I}=k$ and by \Cref{claim:IS:size} we have $B\cap I=\abs{B\setminus N(A_1[\cS_f])}=2^k$.
	Hence, the weight of $I$ is $2^{2k}km + 2^{k}+m$.
\end{proof}

\begin{lemma}\label{lem:IS:SAT}
	If $G$ admits an independent set of weight at least $2^{2k}km + 2^{k}+m$, then $\phi$ is a satisfiable 3-CNF-SAT formula.
\end{lemma}
\begin{proof}
	Assume that $G$ admits an independent set $I$ of maximum weight with $w(I)\geq 2^{2k}km + 2^{k}+m$.
	First, we assume towards a contradiction that there exists $i\in [m]$ and $j\in [k]$ such that $I\cap A_i[\cS_j]=\emptyset$.
	Let $s\in \cS_j$.
	By construction of $G$, we have $N(a^i_s) \subseteq B \cup A_i[\cS_j]  \cup T_i$.
	By \Cref{claim:IS:size}, we have $\abs{N(a_s^i)\cap B} = 2^{2k-1}$.
	Moreover, since $T_i$ induces a triangle in $G$, $I$ contains at most one of vertex in $T_i$.
	We deduce that 
\[ 	w(N(a_s^i)\cap I) \leq 2^{2k-1}+1 < w(a_s^i)=2^{2k}. \]
	Consequently, $I'=(I\setminus N(a_s^i)) \cup \{a_s^i\}$ is an independent set of $G$ with $w(I)<w(I')$, yielding a contradiction with $I$ being of maximum weight.
	
	From now, we assume that, for every $i\in [m]$ and $j\in [k]$, $I$ contains exactly one vertex in $A_i[\cS_j]$.
	Since each $A_i[\cS_j]$ induces a clique, we deduce that for every $i\in [m]$, we have $\abs{A_i\cap I}=k$.
	By \Cref{lemma:IS:interpretation}, there exist $m$ interpretations $f_1,\dots,f_m$ of $\var(\phi)$ such that, for every $i\in [m]$, we have $I\cap A_i =A_i[\cS_{f_i}]$.
	By \Cref{claim:IS:size}, for every $i\in [m]$, we have 
	\begin{equation}\label{eq:IS:size:neighborhood}
		\abs{N(A_i[\cS_{f_i}])\cap B}=2^{2k}-2^{k}.
	\end{equation}	
	Hence, we have $w(I\cap B)=\abs{I\cap B}\leq 2^{k}$.
	Since each $T_i$ induces a triangle in $G$, we have $w(I\cap T_i)=\abs{I\cap T_i}\leq 1$.
	As $w(I)\geq 2^{2k}km + 2^{k} +m$, we deduce that $I$ has exactly $2^{k}$ vertices in $B$ and 1 in each $T_j$ for $j\in [m]$.
	
	Equation~\ref{eq:IS:size:neighborhood} and  $\abs{I\cap B_i}=2^{k}$ imply that the neighborhoods of $I\cap A_1, I\cap A_2,\dots,I\cap A_m$ in $B$ is the same.
	From \Cref{claim:IS:neighborhood:different}, we deduce that $f_1=f_2=\dots=f_m$.
	Since $I$ contains exactly one vertex in each $T_i$, we conclude from \Cref{obs:IS:clause:gadget} that $f_1=\dots=f_m$ satisfies every clause of $\phi$.
\end{proof}

\begin{lemma}\label{lem:IS:rw}
	The linear rank-width of $G$ is at most $2k+4$.
\end{lemma}
\begin{proof}
    For each $i\in [m]$ and $j\in [k]$, let $\sigma(A_i[\cS_j])$ be an arbitrary permutation of $A_i[\cS_j]$ and $\sigma(A_i)$ be the concatenation of $\sigma(A_i[\cS_1]),\dots,\sigma(A_i[\cS_k])$.
	For each $X\in \{T_1,\dots,T_m\}\cup \{B\}$, let $\sigma(X)$ be an arbitrary permutation of $X$. 
	We define the permutation $\sigma$ of $V(G)$ as the concatenation of $\sigma(B),\sigma(A_1), \sigma(T_1), \sigma(A_2)$, $\sigma(T_2),\dots,\sigma(A_m)$ and $\sigma(T_m)$.
	We claim that $\rw(\sigma)\leq 2k+4$.
	
	Let $(X,\comp{X})$ be a cut of $G$ induced by $\sigma$.
	If $X\cap (A_1\cup\dots\cup A_m)=0$, then $\rw(X,\comp{X})$ is at most $\rw(B,A_1\cup\dots\cup A_m)$.
Now, observe that $G[B,A_1\cup\dots\cup A_m]$ is obtained from the universal $2k$-rank cut $R_{2k}$ by removing some vertices $A^{2k}$ and making copies of the ones we do not remove.
	Consequently, we have $\rw(B,A_1\cup\dots\cup A_m) \leq 2k$.
	
	Suppose now that $X\cap (A_1\cup\dots\cup A_m)\neq \emptyset$. Let $i\in [m]$ and $j\in [k]$ be maximum such that $X\cap A_i[\cS_j]\neq \emptyset$.
	Observe that the only edges of $G[X,\comp{X}]$ are (1)~between $X\cap B$ and $\comp{X}\cap (A_{j+1}\cup\dots\cup A_m)$, (2)~between $A_i[\cS_j]\cap X$ and $A_i[\cS_j]\cap \comp{X}$, (3)~between $X\cap (A_i\cup T_i)$ and $\comp{X}\cap T_i$.
	As $G[A_i]$ is a clique, we have $\rw(A_i[\cS_j]\cap X, A_i[\cS_j]\cap \comp{X})\leq 1$.
	Since $\abs{\comp{X}\cap T_i}\leq 3$, we deduce that $\rw(X,\comp{X})$ is at most $2k + 4$.
	As it holds for any prefix $X$ of $\sigma$, we conclude that $\lrw(G) \leq 2k+4$.
\end{proof}

\begin{theorem}\label{thm:IS}
	There is no algorithm solving \textsc{Weighted Independent Set} in time $2^{o(\lrw(G)^{2})} n^{O(1)}$ unless ETH fails.
\end{theorem}
\begin{proof}
	Assume that there exists a $2^{o(\lrw(G)^{2})} n^{O(1)}$ time algorithm for \textsc{Weighted Independent Set}.
	We prove that it implies the existence of a $2^{o(k^2)} n^{O(1)}$ time algorithm for \textsc{3-SAT} where $k^2$ is the number of variables. This will contradict ETH.
	
	Suppose that we are given a 3-SAT formula $\phi$ with $k^2$ variables and $m$ clauses.
	We construct the graph $G$ described above.
	As $G$ has $2^{2k} + (2^kk+3)m$ vertices, we deduce that we can construct $G$ in time $2^{O(k)} m$.
	
	From \Cref{lem:IS:SAT,lem:SAT:IS}, we know that $G$ admits an independent set of weight at least $2^{2k}km + 2^{k} + m$ if and only if $\phi$ is satisfiable.
	By assumption, we can compute an independent set of $G$ in time $2^{o(\lrw(G)^{2})} n^{O(1)}$.
	By \Cref{lem:IS:rw}, the linear rank-width of $G$ is at most $2k+4$.
	Hence, we can decide whether $\phi$ is satisfiable in time $2^{o(k^2)} n^{O(1)}$.
	This contradicts ETH by \Cref{lem:3SAT:square:ETH}.
\end{proof}

\begin{lemma}\label{lem:IS:unweighted}
	Let $G$ be a graph with a weight function $w : V(G)\to \bN$ and $\sigma$ be a linear decomposition of rank-width $w$. 
	We can construct in time $O(\abs{V(G)}  \max_{v\in V(G)} w(v))$ a graph $G'$ and a linear decomposition $\sigma'$ of $G'$ with rank-width at most $w+1$ such that $G$ admits an independent set $I$ of weight at least $W$ iff $G'$ admits a independent set of size at least $W$.
\end{lemma}
\begin{proof}
	We assume without loss of generality that $G$ has no vertex of weight 0 (we can always delete them without changing the weights of the independents sets of $G$). 
	The graph $G'$ is obtained from $G$ by adding iteratively $w(v)-1$ false twins to each vertex $v\in V(G)$. 
	Formally, $G'$ is the graph with vertex set $V(G')=\{v_i\mid v\in V(G) \land i\in [w(v)]\}$ and edge set $\{u_iv_j \mid uv\in E(G) \land i\in [w(u)] \land j\in [w(v)]\}$.
	We construct $\sigma'$ from $\sigma$ by replacing every vertex $v\in V(G)$ by the sequence $(v_1,\dots,v_{w(v)})$.
    Obviously, $G'$ and $\sigma'$ can be constructed in time $O(\abs{V(G)}  \max_{v\in V(G)} w(v))$
	
	Given an independent set $I$ of $G$, it is easy to see that $\{ v_i \mid v\in I \land i\in [w(v)]\}$ is an independent set of $G$' of size $w(I)$.
	On the other hand, for every independent set $I'$ of $G'$, we have $\abs{I'}\leq w(\{v \in V(G) \mid I\cap \{v_1,\dots,v_{w(v)}\}\neq \emptyset\})$.
	
	Let $(A',B')$ be a cut induced by $\sigma'$,  $A=\{v\in V(G)\mid v_1\in A'\}$ and $B=\{v\in V(G)\mid v_{w(v)}\in B'\}$.
	By construction, $(A,V(G)\setminus A)$ is a cut of $G$ induced by $\sigma$ and $B$ is either $V(G)\setminus A$ or $(V(G)\setminus A)\cup \{v\}$ for some vertex $v\in A$ if $v_1\in A'$ and $v_{w(v)}\in B'$.
	Moreover, the adjacency matrix between $A'$ and $B'$ in $G'$ can be obtained from the one between $A$ and $B$ in $G$ by adding copies of rows and columns.
	Since adding a copy of a row or a column does not increase the rank of a matrix, we conclude that $\rw(A',B')\leq \rw(A,B) + \rw(A,\{v\}) \leq \rw(A,V(G)\setminus A) + 1$.
	As $\rw(A,V(G))\leq w$, we deduce that $\rw(A',B')\leq w+1$.
	We conclude that the rank-width of $\sigma'$ is at most $w+1$.
\end{proof}

Theorem~\ref{thm:main} is now a direct consequence of \Cref{thm:IS,lem:IS:unweighted}.

\section{Maximum Induced Matching and Feedback Vertex Set}
\label{sec:MIMFVS}

In this subsection, we prove that our lower bound for \textsc{Independent Set} holds also for \textsc{Maximum Induced Matching} and \textsc{Feedback Vertex Set}. To prove this, we provide a single reduction from \textsc{Independent Set} that works for both problems. 

\begin{lemma}\label{lem:mim:fvs}
	Let $G$ be a graph, $\sigma$ be a linear decomposition of $G$ of rank-width $w$. We can construct in polynomial time a graph $G'$ and a linear decomposition of $G'$ of rank-width at most $w+1$ such that, for every $k\in \bN$, the following properties are equivalent:
	\begin{enumerate}
		\item\label{item:is} $G$ admits an independent set of size $k$.
		\item\label{item:mim} $G'$ admits an induced matching on $k$ edges.
		\item\label{item:fvs} $G$ admits an induced forest with $2k$ vertices.
	\end{enumerate} 
\end{lemma}
\begin{proof}
	Let $G'$ be the graph with vertex set $\{v,\widehat{v} \mid v\in V(G)\}$ and edge set 
	\[ 	\{v\widehat{v}\mid v\in V(G)\}\cup \{uv, \widehat{u} v,u\widehat{v},\widehat{u} \widehat{v} \mid uv\in E(G) \}. \]
	We construct a linear decomposition $\sigma'$ of $G'$ from $\sigma$ by inserting $\widehat{v}$ after $v$ in $\sigma$ for each vertex $v\in V(G)$.
	Obviously, $G'$ and $\sigma'$ can be constructed in polynomial time.
	Since $\widehat{v}$ is a true twin of $v$ in $G'$ for every $v\in V(G)$, we deduce that the rank-width of $\sigma'$ is at most $w+1$.
	
	\subparagraph{($\ref{item:is}\Rightarrow\ref{item:mim}\land \ref{item:fvs}$)} Given an independent set $I$ of $G$, the set of vertices $\{v\widehat{v}\mid v\in I\}$ induces a matching (and a forest) of $G'$.
	
	\subparagraph*{($\ref{item:mim}\Rightarrow\ref{item:is}$)} For $M$ an induced matching $M$ of $G'$, we obtain an independent set of $G$ of size $\abs{M}$ by considering any set of vertices that contains exactly one endpoint in $V(G)$ of each edge in $M$.
	
	\subparagraph{($\ref{item:fvs}\Rightarrow\ref{item:is}$)} Let $F$ be a forest of $G'$ of maximum size with a maximum number of edges in $\widehat{E}\defeq\{v\widehat{v}\mid v\in V(G)\}$.
	Assume towards a contradiction that $F$ admits a connected component $C$ with an edge not in $\widehat{E}$.
	By construction, we deduce that $V(C)$ contain at most one vertices in $\{v,\widehat{v}\}$ for each $v\in V(G)$.
	As $C$ is a tree, there exists $I_C\subseteq C$ such that $2\abs{I_C}\geq \abs{C}$ and $G'[I_C]$ is an independent set.
	Observe that $M=\{v,\widehat{v}\mid v\in I_C \lor \widehat{v}\in I_C\}$ induces a matching of size at least $\abs{C}$.
	It follows that $(V(F)\setminus V(C))\cup M$ induces a forest with at least as much vertices as $F$ and with more edges in $\widehat{E}$, yielding a contradiction.
	Consequently, we have $E(F)\subseteq \widehat{E}$ and we conclude that $V(F)\cap V(G)$ is an independent set of size $\abs{V(F)/2}$.
\end{proof}

The following corollary is a direct consequence of \Cref{thm:main,lem:mim:fvs}.

\begin{corollary}\label{cor:mim}
	There is no algorithm solving \textsc{Maximum Induced Matching} or \textsc{Feedback Vertex Set} in time $2^{o(\rw(G)^{2})} n^{O(1)}$ unless ETH fails.
\end{corollary}

\section{Weighted Dominating Set}
\label{sec:wds}

As for \textsc{Independent Set}, the starting point is an instance $\phi$ of \textsc{3-CNF-SAT} with $m$ clauses $C_1,\dots,C_m$ and a set of $k^2$ variables $\var(\phi)\defeq \{v_{i,j} \mid i\in [k], j\in [k+1,2k]\}$.
Similarly to \textsc{Independent Set}, we construct a graph $G$ with $m$ copies $A_1,\dots,A_m$ of $A^{2k}[\cS]$ and we make sure that for any dominating set of minimum weight $D$, the intersection of $D$ with $A_1,\dots,A_m$ corresponds to $A_1[\cS_f],\dots,A_m[\cS_f]$ for some interpretation $f$ of $\var(\phi)$.
The clause gadget associated with $C_i$ consists of a single vertex adjacent to the vertices of the $i$-th copy of $A^{2k}[\cS]$ that represent the partial interpretations of $\var(\phi)$ satisfying $C_i$.
Thus, if $\phi$ is satisfied by an interpretation $f$, a dominating set including $A_1[\cS_f],\dots,A_m[\cS_f]$ would dominate the vertices associated with the clause gadgets.

The main difference with our reduction for \textsc{Independent Set} is that we need $2(m-1)$ copies $B_1,\widehat{B}_1,\dots,B_{m-1},\widehat{B}_{m-1}$ of $B^{2k}$  to guarantee our equivalence between minimum dominating sets and interpretations of $\var(\phi)$.
Briefly, for each $i\in[m-1]$, $G[A_i,B_i]$ and $G[\widehat{B}_i,A_{i+1}]$ are isomorphic to $R_{2k}[A^{2k}[\cS],B^{2k}]$ and $G[B_i,\widehat{B}_i]$ is an induced matching such that for each set $s\subseteq [2k]$, the two vertices in $B_i\cup \widehat{B}_i$ associated with $s$ are adjacent. See \Cref{fig:dsoverview} for an overview of the reduction.

This path-shaped construction and prohibitive weights on some vertices ensure that $\phi$ is satisfiable by an interpretation $f$ iff the set containing $A_1[\cS_f],\dots,A_m[\cS_f]$ and $B_1\setminus N(A_1[\cS_f]),\dots,B_{m-1}\setminus N(A_m[\cS_f])$ is a dominating set of minimum weight.
Thanks to the induced matchings between $B_i$ and $\widehat{B}_i$, the vertices of $B_i\setminus N(A_i[\cS_f])$ dominate the vertices in $\widehat{B}_i$ that are not dominated by $A_{i+1}[\cS_f]$.

\begin{figure}
	\centering
	\includegraphics[width=0.9\linewidth]{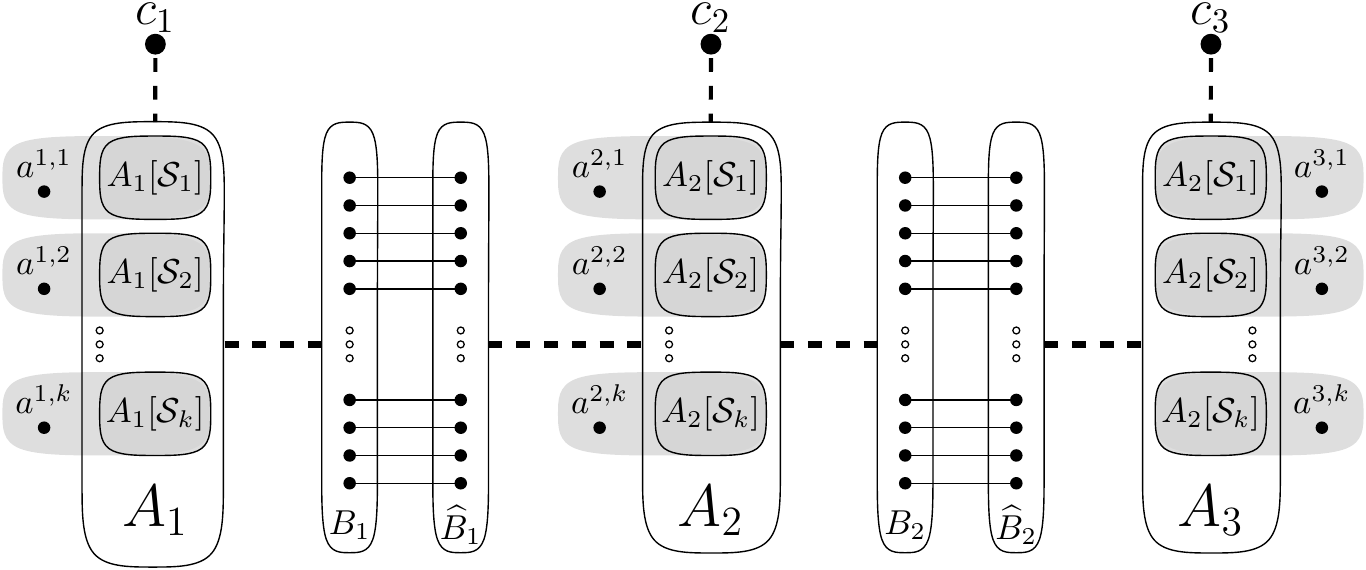}
	\caption{Overview of the reduction for \textsc{Weighted Dominating Set} with $m=3$. The gray areas represents cliques and dotted lines indicates the existence of edges between two sets of vertices.}
	\label{fig:dsoverview}
\end{figure}

\subparagraph{The construction.}  
We construct a graph $G$ as follows.
We create $m$ copies $A_1,\dots,A_m$ of $A^{2k}[\cS]$, for each $i\in [m]$, we have $A_i\defeq \{a^i_s \mid s \in \cS\}$.
For each $i\in[m]$ and $j\in [k]$, we create a vertex $a^{i,j}$ and we add edges so that $A_i[\cS_j] \cup \{a^{i,j}\}$ induces a clique denoted by $K_{i,j}$.
As the vertex $a^{i,j}$ will be only adjacent to the vertex in $A_i[\cS_j]$.
This guarantee that any dominating set of $G$ contains at least one vertex in $K_{i,j}$.


We create $2(m-1)$ copies $B_1,\widehat{B}_1,\dots,B_{m-1},\widehat{B}_{m-1}$ of $B^{2k}$, for each $i\in [m-1]$, we have $B_i\defeq \{b^i_s \mid s \subset [2k]\}$ and $\widehat{B}_i\defeq \{\widehat{b}^i_s \mid s \subseteq [2k]\}$.
For every $i\in [m-1]$, $s\in \cS$ and $t\subseteq [2k]$ such that $\abs{s\cap t}$ is odd, we make (1)~$a^i_s$ adjacent to $b^i_t$, (2)~$\widehat{b}^i_t$ adjacent to $a^{i+1}_s$ and (3)~$b_t^i$ adjacent to $\widehat{b}_t^i$.
 Consequently, $G[A_i,B_i]$ and $G[\widehat{B}_i,A_{i+1}]$ are both isomorphic to $R_{2k}[A^{2k}[\cS],B^{2k}]$ and $G[B_i,\widehat{B}_i]$ is an induced perfect matching.

For every $i\in[m]$ with $\widehat{B}_i=\{\ell_1,\ell_2,\ell_3\}$, we create a vertex $c_i$ adjacent to the vertices in $A_i[\cS_{\ell_1}\cup \cS_{\ell_2}\cup \cS_{\ell_3}]$.

Finally, we define the weight function $w : V(G) \to \bN$ such each vertex $v\in B_1\cup \dots \cup B_{m-1}\cup \{c_i \mid i\in [m]\}$  has weight $w(v)\defeq 1$, each vertex $u\in A_1\cup \dots \cup A_m$ has weight $w(u)\defeq 2^{2k}+2$ and every vertex $x\in \widehat{B}_1 \cup \dots \cup \widehat{B}_{m-1} \cup \{a^{i,j}\mid i\in [m]\land j\in [k]\}$ has weight $w(x)\defeq+\infty$.
The purpose of $w$ is to guarantee that every minimum dominating set of $G$ contains at most $1$ vertices in each $A_i[\cS_j]$ and no vertex in $\widehat{B}_1 \dots \widehat{B}_{m-1} \cup \{a^{i,j}\mid i\in [m]\land j\in [k]\}$.

\begin{lemma}\label{lem:DS:SATtoDS}
	If $\phi$ is satisfiable, then $G$ admits a dominating set of weight $(2^{2k}+2)km + 2^{k}(m-1)$.
\end{lemma}
\begin{proof}
	Suppose that $\phi$ is satisfied by an interpretation $f$.
	Let $D$ be the union of $A_i[\cS_f]$ and $B_j\setminus N( A_{j}[\cS_f])$ for every $i\in [m]$ and $j\in [m-1]$.
	We claim that $D$ is a dominating set of weight $(2^{2k}+2)km + 2^{k}(m-1)$.
	The weight of $D$ is deductible from the following observations:
	\begin{itemize}
		\item For each $i\in [m]$, the weight of each vertex in $A_i$ is $2^{2k}+2$ and $\abs{A_i[\cS_f]}=k$.
		\item The weight of each vertex in $B_1\cup \dots \cup B_{m-1}$ is 1, and by \Cref{obs:sF_k}, we have $\cS_f\in\sF_{2k}$ which implies with \Cref{lem:neighborhood:size} that $\abs{B_j\setminus N( A_{j}[\cS_f])}=2^k$ for each $j\in [m-1]$.
	\end{itemize}
	We conclude that $D$ is a dominating set of $G$ from the following arguments:
	\begin{itemize}
		\item Let $i\in [m]$.
		By definition, $A_i[\cS_f]$ contains one vertex in $A_i[\cS_j]$ for every $j\in[k]$.
		As $K_{i,j}=A_i[\cS_j]\cup a^{i,j}$ is a clique, we deduce that $D$ dominates $A_i$ and $\{a^{i,j} \mid j\in [k]\}$.
		Moreover, $f$ interprets at least one literal $\ell$ of $C_i$ as true.
		Thus, $A_i[\cS_f]\cap A_i[\cS_\ell]\neq \emptyset$ and $c_i$---the vertex representing $C_i$--- has a neighbor in $D$.
		So, $D$ dominates $A_i\cup \{c_i\}\cup \{a^{i,j}\mid j\in [k]\}$ for every $i\in[m]$.
		
		\item Let $j\in [m-1]$.
		As $A_j[\cS_f]$ and $B_j\setminus N( A_{j}[\cS_f])$ are included in $D$, we know that $D$ dominates $B_j$.
		Let $\widehat{b}_s^j$ be a vertex in $\widehat{B}_j$.
		If $b_s^j\in B_j$ is in $D$, then $\widehat{b}_s^j$ is dominated by $D$ as $b_s^j \widehat{b}_s^j$ is an edge of $G$.
		Otherwise, if $b_s^j$ is not in $D$, then $b_s^j$ is adjacent to a vertex $a_t^j$ in $A_j[\cS_f]\subseteq D$.
		In this later case, the vertex $a_t^{j+1}$ belongs to $A_{j+1}[\cS_f] \subseteq D$ and $a_t^{j+1}$ is adjacent to $\widehat{b}_s^j$.
		It follows that $D$ dominates $\widehat{B}_i$.
	\end{itemize}
\end{proof}

\begin{lemma}\label{lem:DS:DStoSAT}
	If $G$ admits a dominating set of weight at most $(2^{2k}+2)km + 2^k(m-1)$, then $\var(\phi)$ is satisfiable.
\end{lemma}
\begin{proof}
	Let $D$ be a dominating set of weight at most $(2^{2k}+2)km + 2^k(m-1)$.
	We use the following claim prove that there exists some interpretation $f$ such that $D$ is the union of $A_1[\cS_f], \dots, A_m[\cS_f]$ and $B_1\setminus N(A_1[\cS_f]), \dots, B_m\setminus N(A_m[\cS_f])$.
	
	\begin{claim}\label{claim:ds:interpretation}
		For every $i\in [m]$, there exists an interpretation $f_i$ of $\var(\phi)$ such that $A_i\cap D = A_i[\cS_f]$.
	\end{claim}
	\begin{proof}
		By \Cref{def:interpretation:sF2k}, it is sufficient to prove that $D$ contains exactly one vertex in $A_i[\cS_j]$ for every $i\in[m]$ and $j\in [k]$.
		Let $i\in [m]$ and $j\in [k]$.
		By construction, we have $N(a^{i,j})=A_i[\cS_j]$.
		Since $a^{i,j}$ must be dominated by $D$ and $w(a^{i,j})=+\infty$, $D$ contain at least one vertex in $A_i[\cS_j]$.
		
		Assume towards a contradiction that $D$ contains two different vertices $u,v$ in $A_i[\cS_j]$.
		By construction, we have $N(u)\setminus N(v) \subseteq  \widehat{B}_{i-1}\cup B_i \cup \{c_i\}$ (we consider that $B_0=\widehat{B}_0=B_m=\emptyset$) and the vertices in $\{c_i\}\cup \widehat{B}_{i-1}\cup B_i$ are dominated by the set $X\defeq B_{i-1}\cup B_i \cup \{c_i\}$.
		Thus, $(D\setminus \{u\})\cup X$ is a dominating set of $G$.
		But, as $w(X)=\abs{X} \leq 2^{2k}+1 < w(u) = 2^{2k}+2$, this contradicts $D$ being a dominating set of minimum weight.
		We conclude that $D\cap A_i = A_i[\cS_{f_i}]$ for some interpretation $f_i$ of $\var(\phi)$.
	\end{proof}
	
	For every $i\in[m-1]$, we denote by $B_i^\star$ the vertices in $B_i$ that are not dominated by $A_i[\cS_{f_i}]$, i.e. $B_i^\star\defeq B_i \setminus N(A_i[\cS_{f_i}])$. Similarly, we denote by $\widehat{B}_i^\star$ the vertices in $\widehat{B}_i$ not dominated by $A_{i+1}[\cS_{f_{i+1}}]$.
	By \Cref{obs:sF_k}, we have $\cS_{f_1},\dots,\cS_{f_m}\in \sF_{2k}$ and from \Cref{lem:neighborhood:size}, we deduce that $\abs{B_i^\star} = \abs{\widehat{B}_i^\star} = 2^{k}.$
	By construction, each $B_i^\star$ is an independent set and $N(B_i^\star)$ is included in $A_i \cup \widehat{B}_i$.
	As $A_i\cap D = A_i[\cS_{f_i}]$ and $w(v)=+\infty$ for every $v\in \widehat{B}_i$, we deduce that $D$ contains $B_i^\star$ for every $i\in [m-1]$.
	
	Observe that $w(A_i[\cS_{f_i}])=(2^{2k}+2)k$ for every $i\in [m]$ and $w(B_j^\star)=\abs{B_j^\star}= 2^{k}$ for each $j\in[m-1]$.
	As $w(D)$ is at most $(2^{2k}+2)km + 2^{k}(m-1)$, we deduce that $D$ is exactly the union of $A_i[\cS_i]$ and $B_j^\star$ for $i\in[m]$ and $j\in [m-1]$.
	
	By definition, for each $i\in [m-1]$, the vertices in $\widehat{B}_i^\star$ are not dominated by $A_{i+1}[\cS_{f_{i+1}}]$ and thus they must be dominated by $B_i^\star$.
	As $G[B_i,\widehat{B}_i]$ is an induced perfect matching with set of edges $\{b_s^i\widehat{b}_s^i\mid s\subseteq [2k]\}$, we deduce that $\widehat{B}_i^\star=\{\widehat{b}_s^i \mid b_s^i \in B_i^\star\}$.
	This implies that $N(A_{i}[\cS_{f_i}])= N(A_{i}[\cS_{f_{i+1}}])$.
	Since $f_i$ and $f_{i+1}$ belong to $\sF_{2k}$, by \Cref{lem:different:neighborhood}, it follows that $f_i=f_{i+1}$.
	We deduce that $f_1=f_2=\dots=f_m$.
	We conclude that the interpretation $f_1=\dots=f_m$ satisfies $\phi$ because for every $i\in [m]$, the vertex $c_i$ representing the clause $C_i$ is dominated by $A_i[\cS_{f_i}]$ and by \Cref{obs:interpretation:literal} it implies that $f_i$ satisfies $C_i$.
\end{proof}

\begin{lemma}\label{lem:DS:rank}
	We can compute in polynomial time a linear decomposition of $G$ with rank-width at most $4k+2$.
\end{lemma}
\begin{proof}
	For every $X\in \{K_{i,j} \mid i\in [m]\land j\in [k]\}$, let $\sigma(X)$ be an arbitrary permutation of $X$.
	For every $i\in [m]$, we define the permutation $\sigma(A_i)$ of $A_i\cup \{c_i\} \cup \{a^{i,j} \mid j\in [k]\}$ as the concatenation of the permutation $(c_i),\sigma(K_{i,1}),\sigma(K_{i,2}),\dots,\sigma(K_{i,k-1})$ and $\sigma(K_{i,k})$.
	Let $(s_1,\dots,s_t)$ be a permutation of $2^{[2k]}$, for every $i\in [m-1]$, we define $\sigma(B_i\cup \widehat{B}_i)$ as the permutation $(b^i_{s_1},\widehat{b}^i_{s_1},b^i_{s_2},\dots,\widehat{b}^{i}_{s_{t-1}},b^{i}_{s_t},\widehat{b}^i_{s_t})$.
	Let $\sigma$ be the concatenation of $\sigma(A_1),\sigma(B_1\cup \widehat{B}_1),\sigma(A_2),\dots,\sigma(B_{m-1}\cup \widehat{B}_{m-1})$ and $\sigma(A_m)$.
	
	Obviously, $\sigma$ is a linear decomposition of $G$ that can be computed in polynomial time.
	We claim that the rank-width of $\sigma$ is at most $4k+2$.
	Let $(X,\comp{X})$ be a cut of $G$ induced by $\sigma$.
	We distinguish the following cases:
	\begin{itemize}
		\item Suppose that there exists $i\in[m]$ such that $X$ intersect $A_i\cup \{c_i\} \cup \{a^{i,j}\mid j\in [k]\}$ but not $B_i$ (we consider that $B_m=\emptyset$).
		The edges of $G[X,\comp{X}]$ belong to the following cuts: (1)~the cut between $\widehat{B}_{i-1}$ and $A_i \cap \comp{X}$, (2)~the cut between $A_i\cap X$ and $B_i$, (3)~the cut between $c_i$ and $A_i \cap \comp{X}$ and (4)~the cut between $K_{i,j}\cap X$ and $K_{i,j}\cap \comp{X}$ with $j=\max\{ \ell \in [k] \mid K_{i,\ell}\cap X\neq \emptyset\}$.
		
		The ranks of the first two cuts are upper bounded by $\rw(A_i,\widehat{B}_{i-1})$ and $\rw(A_i,B_i)$ respectively, since $G[A_i,\widehat{B}_{i-1}]$ and $G[A_i,B_i]$ are isomorphic to $R_{2k}$, these ranks are at most $2k$.
		The rank of the third cut is upper bounded by $1$ since one side consists of a single vertex. 
		Since the fourth cut is a biclique, its rank is at most 1.
		We deduce that $\rw(X,\comp{X})\leq 4k + 2$.
		
		\item Suppose now that there exists $i\in[m-1]$ such that $X$ intersect $B_i\cup \widehat{B}_i$ but not $A_{i+1}$.
		Let $b^i_{s}\in B_i$ be the rightmost vertex in $\sigma$ that belongs to $X$.
		The edges of $G[X,\comp{X}]$ belong to the following cuts: (1)~the cut between $A_i$ and $B_i\cap \comp{X}$, (2)~the cut between $\widehat{B}_i\cap X$ and $A_{i+1}$ and (3)~the cut between $b^i_s$ and $\{\widehat{b}_s^i\}\cap \comp{X}$.
		As argued for the previous case, the ranks of the first two cuts are upper bounded by $2k$ and the rank of the third is upper bounded by one. We conclude that $\rw(X,\comp{X})\leq 4k+1$.
	\end{itemize}
\end{proof}

\begin{theorem}\label{thm:DS}
	There is no algorithm solving \textsc{Weighted Dominating Set} in time $2^{o(\lrw(G)^{2})} n^{O(1)}$ unless ETH fails.
\end{theorem}
\begin{proof}
	Assume that there exists a $2^{o(\lrw(G)^{2})} n^{O(1)}$ time algorithm for \textsc{Weighted Dominating Set}.
	We prove that it implies the existence of a $2^{o(k^2)} n^{O(1)}$ time algorithm for \textsc{3-SAT} where $k^2$ is the number of variables. This will contradict ETH.
	
	Suppose that we are given a 3-SAT formula $\phi$ with $k^2$ variables and $m$ clauses.
	We construct the graph $G$ described above.
	As $G$ has $(2^{k}k+k+1)m + 2^{2k}(m-1)$ vertices, we deduce that we can construct $G$ in time $2^{O(k)} m$.
	
	From \Cref{lem:DS:SATtoDS,lem:DS:DStoSAT}, we know that $G$ admits a dominating set of weight at most $(2^{2k}+2)km + 2^k(m-1)$ iff $\phi$ is satisfiable.
	Thanks to \Cref{lem:DS:rank}, we can compute in polynomial time a linear decomposition of $G$ of rank-width at most $4k+2$.
	By assumption, we can compute a dominating set of minimum weight of $G$ in time $2^{o(k^{2})} n^{O(1)}$.
	Hence, we can decide whether $\phi$ is satisfiable in time $2^{o(k^2)} n^{O(1)}$.
	This contradicts ETH by \Cref{lem:3SAT:square:ETH}.
\end{proof}

\section{Relation of Boolean-width and Rank-width}
\label{sec:rel}
In this section, we show that for every integer $k \ge 1$ there exists a graph with rank-width at most $2k+1$ and Boolean-width at least $k(k-3)/6 = \Omega(k^2)$.
This answers negatively the question asked by Bui-Xuan, Telle, and Vatshelle in 2011 on whether Boolean-width is subquadratic in rank-width on all graphs~\cite{DBLP:journals/tcs/Bui-XuanTV11}.

Let $k$ be a positive integer.
We construct a graph $G$ in a similar, but slightly different way than in Section~\ref{sec:IS}.
We again start from the universal $2k$-rank cut $R_{2k}$, and remove all vertices in the side $A^{2k}$ that are not in the set $A^{2k}[\cS] = A^{2k}[\{s \subseteq [2k] \mid |s \cap [k]| = 1\}]$, and then for each $i \in [k]$ make the sets $A^{2k}[\cS_i] = A^{2k}[\{s \in \cS \mid s \cap [k] = \{i\}\}]$ into cliques.
Then, we make $k^2$ copies $B_1,\ldots,B_{k^2}$ of the other side $B^{2k}$ of the cut, and make each of them a clique.
Now, for each $i \in [k^2]$, the bipartite graph $G[A^{2k}[\cS], B_i]$ is isomorphic to $R_{2k}[A^{2k}[\cS], B^{2k}]$.
We arrive to a construction resembling Figure~\ref{fig:isoverview}, but this time $A^{2k}[\cS]$ is the ``center'' and $B_i$:s are the ``leaves''.

The following lemma encapsulates a standard argument on branch decompositions that we need to lower bound the Boolean-width.

\begin{lemma}
\label{lem:balsep}
Let $f$ be a function $f : 2^V \rightarrow \mathbb{Z}_{\ge 0}$ and $X \subseteq V$ with $|X| \ge 2$.
Any branch decomposition of $f$ has an edge that displays a bipartition $(L,R)$ of $V$ so that $|L \cap X| \ge |X|/3$ and $|R \cap X| \ge |X|/3$.
\end{lemma}
\begin{proof}
For every edge $uv$ of the branch decomposition displaying a cut $(L_u,R_v)$, direct $uv$ from $u$ to $v$ if $|R_v \cap X| > |L_u \cap X|$ and from $v$ to $u$ if $|L_u \cap X| > |R_v \cap X|$.
If $|L_u \cap X| = |R_v \cap X|$, then we are done by taking this edge, so assume that every edge is directed in either direction.
Now, because the branch decomposition is a tree, by following the directed edges we end up in a node so that all of its incident edges are directed towards it.
Because $|X| \ge 2$, this node cannot be a leaf, so it corresponds to an internal node displaying a tripartition $(C_1,C_2,C_3)$ of $V$.
Because all edges are directed towards this node, it holds that $|C_i \cap X| < |X|/2$ for all $i \in [3]$.
In particular, if we permute the tripartition so that $C_1$ maximizes the intersection $|C_1 \cap X|$, then $|X|/3 \le |C_1 \cap X| < |X|/2$, which implies that the edge corresponding to $C_1$ displays the cut $(C_1, C_2 \cup C_3)$ with $|C_1 \cap X| \ge |X|/3$ and $|(C_2 \cup C_3) \cap X| > |X|/2$.
\end{proof}

We then prove the lower bound on the Boolean-width of $G$.

\begin{lemma}
The graph $G$ has Boolean-width at least $k(k-3)/6$.
\end{lemma}
\begin{proof}
Consider a branch decomposition of $V(G)$ that minimizes the Boolean-width.
By Lemma~\ref{lem:balsep}, there exists an edge of the decomposition displaying a bipartition $(L,R)$ of $V(G)$ so that $|A^{2k}[\cS] \cap L| \ge |A^{2k}[\cS]|/3$ and $|A^{2k}[\cS] \cap R| \ge |A^{2k}[\cS]|/3$.
We say that $(L,R)$ cuts a set $B_i$ if $B_i \cap L \neq \emptyset$ and $B_i \cap R \neq \emptyset$.
Suppose $(L, R)$ cuts all $k^2$ sets $B_i$.
Now, because each $B_i$ is a clique and non-adjacent to other $B_i$:s, there exists an induced matching with $k^2$ edges from $L$ to $R$.
This implies that $(L,R)$ has Boolean-width at least $k^2$.
It remains to consider the case that there is a set $B_i$ so that $(L,R)$ does not cut $B_i$.
Without loss of generality, assume that $B_i \subseteq R$.

Now, it remains to lower bound the number of neighborhoods from $A^{2k}[\cS] \cap L$ to $B_i$.
By \cref{lemma:IS:interpretation}, for any selection $\cS' \subseteq \cS$ so that $A^{2k}[\cS']$ is an independent set it holds that $\cS' \subseteq \sF_{2k}$, and therefore by \cref{lem:different:neighborhood} any two different such selections have different neighborhoods into $B_i$.
It follows that the number of neighborhoods from $A^{2k}[\cS] \cap L$ to $B_i$ is at least the number of independent sets in $A^{2k}[\cS] \cap L$, so it remains to lower bound the number of independent sets in  $A^{2k}[\cS] \cap L$.
Recall that the only edges in $A^{2k}[\cS]$ are induced by the disjoint cliques $A^{2k}[\cS_i]$ for $i \in [k]$.
Because $|A^{2k}[\cS] \cap L| \ge |A^{2k}[\cS]|/3$, and $|\cS_i| = 2^k$ for each $i \in [k]$, there are at least $k/6$ indices $i$ so that $|A^{2k}[\cS_i] \cap L| \ge 2^k/6$.
By selecting a single vertex from each such $A^{2k}[\cS_i] \cap L$, we can construct at least 
$(2^k/6)^{k/6} \ge 2^{(k-3)k/6}$ independent sets in $A^{2k}[\cS] \cap L$.
It follows that $\bw(A^{2k}[\cS] \cap L, B_i) \ge k(k-3)/6$, which implies that $\bw(L, R) \ge k(k-3)/6$, which implies that the Boolean-width of the branch decomposition is at least $k(k-3)/6$.
\end{proof}

We then prove the upper bound on the rank-width of $G$.

\begin{lemma}
The rank-width of $G$ is at most $2k+1$.
\end{lemma}
\begin{proof}
We prove that the linear rank-width of $G$ is at most $2k+1$.
For each $i \in [k]$ let $\sigma(A^{2k}[\cS_i])$ be an arbitrary permutation of $A^{2k}[\cS_i]$ and for each $j \in [k^2]$ let $\sigma(B_j)$ be an arbitrary permutation of $B_j$.
We create a permutation $\sigma$ of $V(G)$ as the concatenation \[\sigma(A^{2k}[\cS_1]), \ldots, \sigma(A^{2k}[\cS_k]), \sigma(B_1), \ldots, \sigma(B_{k^2}).\]

Let $(L,R)$ be a cut of $V(G)$ induced by this permutation.
First, consider the case that $(L,R)$ cuts a set $A^{2k}[\cS_i]$ for some $i \in [k]$, i.e., $L \cap A^{2k}[\cS_i] \neq \emptyset$ and $R \cap A^{2k}[\cS_i] \neq \emptyset$.
Now by \cref{lem:cutcomb} and the fact that each $B_j$ has identical neighborhood to $A^{2k}[\cS]$ we have that
\[\rw(L,R) \le \rw(L \cap A^{2k}[\cS_i], R \cap A^{2k}[\cS_i]) + \rw(L \cap A^{2k}[\cS], B_1) \le 1 + \rw(A^{2k}[\cS], B_1) \le 1+2k.\]
Then, consider the case that $(L,R)$ cuts a set $B_j$ for some $j \in [k^2]$.
Now we have that
\[\rw(L,R) \le \rw(A^{2k}[\cS], B_{k^2} \cap R) + \rw(B_j \cap L, B_j \cap R) \le \rw(A^{2k}[\cS], B_{k^2}) + 1 \le 2k+1.\]
In the final case when $(L,R)$ cuts no set $A^{2k}[\cS_i]$ and no set $B_j$ we have that $\rw(L,R) \le \rw(A^{2k}[\cS], B_{k^2}) \le 2k$.
\end{proof}

\section{Concluding Remarks}
\label{sec:conc}
We showed the first ETH-tight lower bounds for problems with time complexity $2^{O(\rw^2)} n^{O(1)}$ parameterized by rank-width $\rw$.
In particular, we showed that algorithms with such time complexity are optimal for \textsc{Independent Set}, \textsc{Weighted Dominating Set}, \textsc{Maximum Induced Matching}, and \textsc{Feedback Vertex Set}.

We hope the tools designed in this paper could be used to design tight lower bounds for more problems parameterized by rank-width.
In particular, is the $2^{O(\rw^2)} n^{O(1)}$ time algorithm in \cite{Bui-XuanTV10} for (unweighted) \textsc{Dominating Set} optimal under ETH? What about the $2^{O(q  \rw^2)} n^{O(1)}$ time algorithm in \cite{Bui-XuanTV10} for \textsc{$q$-Coloring} (even when $q=3$)?
Finally, one could also explore the optimality of XP algorithms parameterized by rank-width such as the $n^{2^{O(\rw)}}$ time algorithm in \cite{DBLP:journals/ejc/GanianHO13} for \textsc{Chromatic Number}. For the clique-width parameterization this was solved in~\cite{DBLP:journals/talg/FominGLSZ19}.

\bibliographystyle{plainurl}
\bibliography{biblio}

\begin{thebibliography}{10}

\bibitem{DBLP:journals/dam/ArnborgP89}
Stefan Arnborg and Andrzej Proskurowski.
\newblock Linear time algorithms for {NP}-hard problems restricted to partial
  k-trees.
\newblock {\em Discret. Appl. Math.}, 23(1):11--24, 1989.
\newblock \href {https://doi.org/10.1016/0166-218X(89)90031-0}
  {\path{doi:10.1016/0166-218X(89)90031-0}}.

\bibitem{DBLP:journals/algorithmica/BelmonteS21}
R{\'{e}}my Belmonte and Ignasi Sau.
\newblock On the complexity of finding large odd induced subgraphs and odd
  colorings.
\newblock {\em Algorithmica}, 83(8):2351--2373, 2021.
\newblock \href {https://doi.org/10.1007/s00453-021-00830-x}
  {\path{doi:10.1007/s00453-021-00830-x}}.

\bibitem{BergougnouxK21}
Benjamin Bergougnoux and Mamadou~Moustapha Kant{\'{e}}.
\newblock More applications of the d-neighbor equivalence: Acyclicity and
  connectivity constraints.
\newblock {\em {SIAM} J. Discret. Math.}, 35(3):1881--1926, 2021.
\newblock \href {https://doi.org/10.1137/20M1350571}
  {\path{doi:10.1137/20M1350571}}.

\bibitem{DBLP:journals/algorithmica/BergougnouxKK20}
Benjamin Bergougnoux, Mamadou~Moustapha Kant{\'{e}}, and O{-}joung Kwon.
\newblock An optimal {XP} algorithm for hamiltonian cycle on graphs of bounded
  clique-width.
\newblock {\em Algorithmica}, 82(6):1654--1674, 2020.
\newblock \href {https://doi.org/10.1007/s00453-019-00663-9}
  {\path{doi:10.1007/s00453-019-00663-9}}.

\bibitem{DBLP:journals/tcs/BroersmaGP13}
Hajo Broersma, Petr~A. Golovach, and Viresh Patel.
\newblock Tight complexity bounds for {FPT} subgraph problems parameterized by
  the clique-width.
\newblock {\em Theor. Comput. Sci.}, 485:69--84, 2013.
\newblock \href {https://doi.org/10.1016/j.tcs.2013.03.008}
  {\path{doi:10.1016/j.tcs.2013.03.008}}.

\bibitem{Bui-XuanTV10}
Binh{-}Minh Bui{-}Xuan, Jan~Arne Telle, and Martin Vatshelle.
\newblock H-join decomposable graphs and algorithms with runtime single
  exponential in rankwidth.
\newblock {\em Discret. Appl. Math.}, 158(7):809--819, 2010.
\newblock \href {https://doi.org/10.1016/j.dam.2009.09.009}
  {\path{doi:10.1016/j.dam.2009.09.009}}.

\bibitem{DBLP:journals/tcs/Bui-XuanTV11}
Binh{-}Minh Bui{-}Xuan, Jan~Arne Telle, and Martin Vatshelle.
\newblock Boolean-width of graphs.
\newblock {\em Theor. Comput. Sci.}, 412(39):5187--5204, 2011.
\newblock \href {https://doi.org/10.1016/j.tcs.2011.05.022}
  {\path{doi:10.1016/j.tcs.2011.05.022}}.

\bibitem{Bui-XuanTV13}
Binh{-}Minh Bui{-}Xuan, Jan~Arne Telle, and Martin Vatshelle.
\newblock Fast dynamic programming for locally checkable vertex subset and
  vertex partitioning problems.
\newblock {\em Theor. Comput. Sci.}, 511:66--76, 2013.
\newblock \href {https://doi.org/10.1016/j.tcs.2013.01.009}
  {\path{doi:10.1016/j.tcs.2013.01.009}}.

\bibitem{DBLP:journals/jcss/CourcelleER93}
Bruno Courcelle, Joost Engelfriet, and Grzegorz Rozenberg.
\newblock Handle-rewriting hypergraph grammars.
\newblock {\em J. Comput. Syst. Sci.}, 46(2):218--270, 1993.
\newblock \href {https://doi.org/10.1016/0022-0000(93)90004-G}
  {\path{doi:10.1016/0022-0000(93)90004-G}}.

\bibitem{DBLP:journals/mst/CourcelleMR00}
Bruno Courcelle, Johann~A. Makowsky, and Udi Rotics.
\newblock Linear time solvable optimization problems on graphs of bounded
  clique-width.
\newblock {\em Theory Comput. Syst.}, 33(2):125--150, 2000.
\newblock \href {https://doi.org/10.1007/s002249910009}
  {\path{doi:10.1007/s002249910009}}.

\bibitem{DBLP:journals/jacm/CyganKN18}
Marek Cygan, Stefan Kratsch, and Jesper Nederlof.
\newblock Fast hamiltonicity checking via bases of perfect matchings.
\newblock {\em J. {ACM}}, 65(3):12:1--12:46, 2018.
\newblock \href {https://doi.org/10.1145/3148227} {\path{doi:10.1145/3148227}}.

\bibitem{DBLP:journals/talg/CyganNPPRW22}
Marek Cygan, Jesper Nederlof, Marcin Pilipczuk, Michal Pilipczuk, Johan M.~M.
  van Rooij, and Jakub~Onufry Wojtaszczyk.
\newblock Solving connectivity problems parameterized by treewidth in single
  exponential time.
\newblock {\em {ACM} Trans. Algorithms}, 18(2):17:1--17:31, 2022.
\newblock \href {https://doi.org/10.1145/3506707} {\path{doi:10.1145/3506707}}.

\bibitem{Diestel12}
Reinhard Diestel.
\newblock {\em Graph Theory, 4th Edition}, volume 173 of {\em Graduate texts in
  mathematics}.
\newblock Springer, 2012.

\bibitem{DBLP:journals/siamdm/FellowsRRS09}
Michael~R. Fellows, Frances~A. Rosamond, Udi Rotics, and Stefan Szeider.
\newblock Clique-width is {NP}-complete.
\newblock {\em {SIAM} J. Discret. Math.}, 23(2):909--939, 2009.
\newblock \href {https://doi.org/10.1137/070687256}
  {\path{doi:10.1137/070687256}}.

\bibitem{DBLP:journals/siamcomp/FominGLS14}
Fedor~V. Fomin, Petr~A. Golovach, Daniel Lokshtanov, and Saket Saurabh.
\newblock Almost optimal lower bounds for problems parameterized by
  clique-width.
\newblock {\em {SIAM} J. Comput.}, 43(5):1541--1563, 2014.
\newblock \href {https://doi.org/10.1137/130910932}
  {\path{doi:10.1137/130910932}}.

\bibitem{DBLP:journals/talg/FominGLSZ19}
Fedor~V. Fomin, Petr~A. Golovach, Daniel Lokshtanov, Saket Saurabh, and Meirav
  Zehavi.
\newblock Clique-width {III:} {H}amiltonian cycle and the odd case of graph
  coloring.
\newblock {\em {ACM} Trans. Algorithms}, 15(1):9:1--9:27, 2019.
\newblock \href {https://doi.org/10.1145/3280824} {\path{doi:10.1145/3280824}}.

\bibitem{DBLP:conf/stoc/FominK22}
Fedor~V. Fomin and Tuukka Korhonen.
\newblock Fast {FPT}-approximation of branchwidth.
\newblock In Stefano Leonardi and Anupam Gupta, editors, {\em {STOC} '22: 54th
  Annual {ACM} {SIGACT} Symposium on Theory of Computing, Rome, Italy, June 20
  - 24, 2022}, pages 886--899. {ACM}, 2022.
\newblock \href {https://doi.org/10.1145/3519935.3519996}
  {\path{doi:10.1145/3519935.3519996}}.

\bibitem{GanianH10}
Robert Ganian and Petr Hlinen{\'{y}}.
\newblock On parse trees and myhill-nerode-type tools for handling graphs of
  bounded rank-width.
\newblock {\em Discret. Appl. Math.}, 158(7):851--867, 2010.
\newblock \href {https://doi.org/10.1016/j.dam.2009.10.018}
  {\path{doi:10.1016/j.dam.2009.10.018}}.

\bibitem{DBLP:journals/ejc/GanianHO13}
Robert Ganian, Petr Hlinen{\'{y}}, and Jan Obdrz{\'{a}}lek.
\newblock A unified approach to polynomial algorithms on graphs of bounded
  (bi-)rank-width.
\newblock {\em Eur. J. Comb.}, 34(3):680--701, 2013.
\newblock \href {https://doi.org/10.1016/j.ejc.2012.07.024}
  {\path{doi:10.1016/j.ejc.2012.07.024}}.

\bibitem{DBLP:conf/stacs/GroenlandMNS22}
Carla Groenland, Isja Mannens, Jesper Nederlof, and Krisztina Szil{\'{a}}gyi.
\newblock Tight bounds for counting colorings and connected edge sets
  parameterized by cutwidth.
\newblock In Petra Berenbrink and Benjamin Monmege, editors, {\em 39th
  International Symposium on Theoretical Aspects of Computer Science, {STACS}
  2022, March 15-18, 2022, Marseille, France (Virtual Conference)}, volume 219
  of {\em LIPIcs}, pages 36:1--36:20. Schloss Dagstuhl - Leibniz-Zentrum
  f{\"{u}}r Informatik, 2022.
\newblock \href {https://doi.org/10.4230/LIPIcs.STACS.2022.36}
  {\path{doi:10.4230/LIPIcs.STACS.2022.36}}.

\bibitem{DBLP:journals/siamcomp/HlinenyO08}
Petr Hlinen{\'{y}} and Sang{-}il Oum.
\newblock Finding branch-decompositions and rank-decompositions.
\newblock {\em {SIAM} J. Comput.}, 38(3):1012--1032, 2008.
\newblock \href {https://doi.org/10.1137/070685920}
  {\path{doi:10.1137/070685920}}.

\bibitem{ImpagliazzoP01}
Russell Impagliazzo and Ramamohan Paturi.
\newblock On the complexity of $k$-{SAT}.
\newblock {\em J. Comput. Syst. Sci.}, 62(2):367--375, 2001.
\newblock \href {https://doi.org/10.1006/jcss.2000.1727}
  {\path{doi:10.1006/jcss.2000.1727}}.

\bibitem{DBLP:journals/tcs/JansenN19}
Bart M.~P. Jansen and Jesper Nederlof.
\newblock Computing the chromatic number using graph decompositions via matrix
  rank.
\newblock {\em Theor. Comput. Sci.}, 795:520--539, 2019.
\newblock \href {https://doi.org/10.1016/j.tcs.2019.08.006}
  {\path{doi:10.1016/j.tcs.2019.08.006}}.

\bibitem{DBLP:journals/siamdm/JeongKO21}
Jisu Jeong, Eun~Jung Kim, and Sang{-}il Oum.
\newblock Finding branch-decompositions of matroids, hypergraphs, and more.
\newblock {\em {SIAM} J. Discret. Math.}, 35(4):2544--2617, 2021.
\newblock \href {https://doi.org/10.1137/19M1285895}
  {\path{doi:10.1137/19M1285895}}.

\bibitem{DBLP:journals/siamdm/Lampis20}
Michael Lampis.
\newblock Finer tight bounds for coloring on clique-width.
\newblock {\em {SIAM} J. Discret. Math.}, 34(3):1538--1558, 2020.
\newblock \href {https://doi.org/10.1137/19M1280326}
  {\path{doi:10.1137/19M1280326}}.

\bibitem{DBLP:journals/talg/LokshtanovMS18}
Daniel Lokshtanov, D{\'{a}}niel Marx, and Saket Saurabh.
\newblock Known algorithms on graphs of bounded treewidth are probably optimal.
\newblock {\em {ACM} Trans. Algorithms}, 14(2):13:1--13:30, 2018.
\newblock \href {https://doi.org/10.1145/3170442} {\path{doi:10.1145/3170442}}.

\bibitem{DBLP:journals/talg/Oum08}
Sang{-}il Oum.
\newblock Approximating rank-width and clique-width quickly.
\newblock {\em {ACM} Trans. Algorithms}, 5(1):10:1--10:20, 2008.
\newblock \href {https://doi.org/10.1145/1435375.1435385}
  {\path{doi:10.1145/1435375.1435385}}.

\bibitem{DBLP:journals/tcs/OumSV14}
Sang{-}il Oum, Sigve~Hortemo S{\ae}ther, and Martin Vatshelle.
\newblock Faster algorithms for vertex partitioning problems parameterized by
  clique-width.
\newblock {\em Theor. Comput. Sci.}, 535:16--24, 2014.
\newblock \href {https://doi.org/10.1016/j.tcs.2014.03.024}
  {\path{doi:10.1016/j.tcs.2014.03.024}}.

\bibitem{DBLP:journals/jct/OumS06}
Sang{-}il Oum and Paul~D. Seymour.
\newblock Approximating clique-width and branch-width.
\newblock {\em J. Comb. Theory, Ser. {B}}, 96(4):514--528, 2006.
\newblock \href {https://doi.org/10.1016/j.jctb.2005.10.006}
  {\path{doi:10.1016/j.jctb.2005.10.006}}.

\bibitem{vatshelle:thesis}
Martin Vatshelle.
\newblock {\em New {W}idth {P}arameters of {G}raphs}.
\newblock PhD thesis, University of Bergen, Norway, 2012.

\end{thebibliography}

\end{document}